\newif\ifsingle
\singlefalse
\ifsingle
\documentclass[12pt, journal, onecolumn, draftclsnofoot]{IEEEtran}
\else
\documentclass[10pt, journal]{IEEEtran}
\fi

\usepackage{epsfig}
\usepackage{epstopdf}
\usepackage{times}
\usepackage{algorithm}
\usepackage{algorithmic}
\usepackage{amssymb}

\usepackage{amsmath}
\usepackage[makeroom]{cancel}
\usepackage{subfigure}
\usepackage{multirow}
\usepackage{url}
\usepackage{eso-pic}
\usepackage{lipsum}
\usepackage{soul}
\usepackage[
  pass,% keep layout unchanged 
]{geometry}

\newcommand{\vs}[1]{{\color{green}\bf[VS: #1]}}

\usepackage{array} 
\def\ignore#1\endignore{}
\newcolumntype{h}{@{}>{\ignore}l<{\endignore}} %% blind out some colums

\newcolumntype{x}[1]{%
>{\centering\hspace{0pt}}p{#1}}%

\usepackage{latexsym}
\usepackage{multicol}
\usepackage{eurosym}
\usepackage{amsthm}
\usepackage{xspace}
\usepackage{pdfpages}
\usepackage{verbatim}

\newtheorem{proposition}{Proposition}
\newtheorem{corollary}{Corollary}

\newtheorem{definition}{Definition}
\newtheorem{remark}{Remark}

\makeatletter
\def\@copyrightspace{\relax}
\makeatother

\begin{document}
%\title{The Pursuit of Efficiency: \\ When Stochastic ABSF Brings Fairness\\ in D2D-enabled 5G Networks}
%\title{The Pursuit of Efficiency: When Stochastic ABSF Brings Fairness in D2D-enabled 5G Networks}
%\title{Trading off Performance for Efficiency in D2D-enabled 5G Networks}
%\title{Joint Optimization of Interference Coordination and Opportunistic D2D Forwarding in 5G Networks}
%\title{Stochastic Interference Orchestration in D2D-enabled Mobile Networks
\title{Fair Stochastic Interference Orchestration with Cellular Throughput Boosted via Outband Sidelinks 
%\\\vm{(Alternative title, with subtitle)\\Repurposing eICIC and D2D for wireless cells\\ 
%\huge A new stochastic method for the evaluation of performance and the optimization of capacity}
}

\author{
Christian Vitale\IEEEauthorrefmark{1},
Vincenzo Sciancalepore\IEEEauthorrefmark{2},
Vincenzo Mancuso\IEEEauthorrefmark{3}%\thanks{V. Sciancalepore and C. Vitale contributed equally to this work.}
\\
\IEEEauthorrefmark{1} Politecnico di Torino \quad\IEEEauthorrefmark{2} NEC Laboratories Europe GmbH \quad\IEEEauthorrefmark{3} IMDEA Networks Institute
%\vspace{-1em}
}

\ifsingle
\author{
Vincenzo Sciancalepore\IEEEauthorrefmark{1},
Christian Vitale\IEEEauthorrefmark{2}\IEEEauthorrefmark{3},
Vincenzo Mancuso\IEEEauthorrefmark{2}%\thanks{V. Sciancalepore and C. Vitale contributed equally to this work.}
\\
$\!\!$\IEEEauthorrefmark{1}NEC Europe Ltd. $\;$\IEEEauthorrefmark{2}IMDEA Networks Institute $\;$\IEEEauthorrefmark{3}Universidad Carlos III de Madrid$\!\!$
%\vspace{-1em}
}
\fi

\maketitle
\psfull

\setlength{\textfloatsep}{5pt}% Remove \textfloatsep

%% Abstract
\begin{abstract}

Time-domain Inter-Cell Interference Coordination (ICIC) is recognized as the main driver towards efficient and effective ultra-dense network deployments.  Almost Blank Subframe (ABS)---as key-example of ICIC---has been recently standardized so as to achieve high spectral efficiency. 

As we show in this article, adopting ABS implies non-trivial complexity to be effective in multicellular environments with heterogeneous cell coverage and user density. Nonetheless, no fairness determinism is guaranteed by ICIC and ABS in particular. Instead, we analytically show that a compound exploitation of ABS with outband {\it sidelinks} used for Device-to-Device (D2D) communications on unlicensed bands not only allows to abate the complexity of operating ABS, but also results in unexpectedly high levels of fairness.

%We provide a deep theoretical analysis of ABS when the millimeter-wave technology is adopted as outband D2D technology. 
Based on the analysis, we formulate a convex optimization problem to {\it stochastically} make ABS decisions while providing proportional fairness guarantees.
Our results prove that, compared to a legacy system, stochastically orchestration of ABS largely boosts fairness while retaining a notable throughput gain offered by mmWave outband sidelinks used for relay. 
%\vm{citare qualche numero a proposito del guadagno}

\end{abstract}
%\vspace{-2mm}

\thispagestyle{empty}	%remove number on the first page

\section{Introduction}
  
\IEEEPARstart{T}{he} novel radio access network (RAN) design brings a number of innovations while opening new challenges because of the data demands exponential increase and of the drastic densification of wireless network access points~\cite{5Gnetworks}. In the business game of a massive deployment of advanced network systems, the verticals segment such as Virtual Mobile Network Operators (MNVOs) and Over-the-top applications (OTTs) are dictating the terms of efficiency and ultra-reliability aspects that, in turn, result in a revolutionary way of conceiving the communications orchestration.

Inter-Cell Interference Coordination (ICIC) and Device-to-Device (D2D) communications represent two key-elements along this innovation roadmap towards advanced RAN deployments%~\cite{xxx}
. Therefore, in this work we focus on their compound performance looking at a practical orchestration with a purely stochastic approach, which, as we show, results in low operational complexity.

Novel enhanced ICIC (eICIC) schemes inherit most of existing approaches for, e.g., beamstearing~\cite{coordinatedBeam_twc}, power control~\cite{power_control_d2d} and massive MIMO%\cite{XX}
, shedding light on implementable and feasible solutions, such as the time-domain traffic scheduling. This paves the road towards a scalable, flexible and high-efficient RAN solution able to accommodate new impelling vertical service requirements. %In this context, a practical and already standardized ICIC technique, namely {\it almost-blank subframe} (often referred to as ABSF or ABS), allows to mute ({\it blank}) a fixed portion of radio frames and thus to reduce the interference caused in given time-slots while improving the channel quality for over-surrounded transmissions points~\cite{feicic_TWC}. It is worth mentioning that some physical resource elements, such as the cell-specific reference symbols (CRS), must be anyway active to provide channel pilot measurements and estimations, hence the term {\it almost blank} in the ABSF acronym. Nevertheless, the additional interference generated by these physical resource elements can be easily handled by interference cancellation techniques (see, e.g., \cite{uplink_cancellation_twc16} and \cite{cancellation_jsac14}).

As our \underline{first main contribution} in this article, we prove that it might be challenging to mark eICIC techniques as a throughput-guaranteeing solution. This is the consequence of imposing ({\it almost}) interference-free transmissions at the expenses of drastically limiting the number of transmissions in neighboring cells. Indeed, eICIC are PHY-layer techniques wherein throughput as well as fairness are naturally out of design scope.
Nevertheless, we show that eICIC offers the possibility to schedule base station activities with little overhead and low computational complexity, hence it can be turned to further enforce fairness. %\st{among cells, though at the expenses of system throughput. Indeed, we found the use of ABSF to achieve fairness quite intuitive, even though it was not designed for fairness purposes.}

If eICIC is properly adopted to provide %\st{transmission efficiency and} 
fairness but not system throughput guarantees, \emph{how is then possible to boost throughput performance?}
The answer lies on a number of innovative schemes proposed for enhancing RAN
%~\cite{xxx}
, leading to our \underline{second main contribution}: the introduction of wideband D2D to relay cellular traffic while being jointly orchestrated with the eICIC solution. 
We do not propose to use {\it inband} sidelinks (i.e., D2D links operated in addition to base station links over the same band) since they would compete for resource access with normal direct cellular links between the base station (gNB in 3GPP jargon indicating a generalized node B) and the user device (UE in 3GPP)~\cite{ed2d}. 
Instead, we propose {\it opportunistic relay} with sidelinks operated over unlicensed spectrum. By adopting directional and electronically steerable antennas for transmissions within a few tens of meters%~\cite{xxx}
, the resulting sidelinks are practically interference-free. This allows to re-think the RAN evolution from a different perspective: the user device. 
%This allows to re-think the RAN evolution including not only cell operators and infrastructure providers but also user devices in the loop: 
Specifically, a mobile node with excellent cellular capabilities can accept to relay traffic for its neighbors momentarily whereas other nodes experience poor channel qualities resulting in a better utilization of cellular radio resources and turning into reduced airtime utilization in the cell. 
Therefore, D2D relay might be  implemented in cellular networks to complement time-blanking techniques without giving up on throughput. Of course, the presence of groups of users leaning toward cooperation is key for the success of opportunistic relay. It could be seconded by widespread social network behaviors that naturally lead to the formation of user groups. For instance, people partake in sharing information and personal contents using short-range communications such as IEEE 802.11-based technologies~\cite{JSAC15_SGBP}, which have been also proposed to offload base station traffic using the D2D paradigm~\cite{asadi2014dronee}. 

%With the stochastic approach here described, we jointly orchestrate the operation of ABSF and D2D to coordinate inter-cell interference while achieving both high throughput and high fairness. 
%Deploying opportunistic D2D relay in an effective manner is challenging in a cellular context because the rapid densification of points of access makes the problem of guaranteeing high-quality links unaffordable with the resources of legacy D2D deployment scenarios. However, a technological solution is ready to serve the next generation of D2D relay communications: millimeter waves communications (mmWave) facilities targeted by the IEEE 802.11ad standard.
AS concerns the technology for implementing outband sidelinks, we consider millimeter waves communications (mmWave), which is also under consideration in 3GPP for 5G networks. Besides, commercial smartphones already implement mmWave protocols under the IEEE 802.11ad standard framework (a.k.a. WiGig).
This new technology achieves virtually unlimited speeds with respect to the cellular capacity, thus going with D2D relay groups one step beyond the classical D2D paradigm. In the rest of the manuscript, we use the term mmD2D to refer to mmWave outband D2D.  
%{\color{red}While WiGig has progressed fast in the last years, no off-the-shelf RAN solution is currently available for introducing the mm-waves communications paradigm as an interference-free access medium. Nevertheless, ABSF will still make the difference even when the transmission interference problem will be partially alleviated (but not completely overcome) with directional mm-waves connections, as they might affect the same spatial area.}
%

%\vs{I would totally remove this paragraph. The compound impact of ABSF and eD2D relay consists in regulating the number of simultaneous cellular transmissions while at the same time reducing the presence of vulnerable users experiencing poor channel conditions. Hence, fairness and transmission efficiency can be jointly targeted by simultaneously controlling inter-cell activity (with ABSF) and intra-cell packet relay (with eD2D) within a novel unified framework.}

%This is possible thanks to the adoption of D2D in cellular networks, so that a clear {\it D2D trademark} can be spotted in the described scenario.

To fully understand the potentials of eICIC and mmD2D in combination, and to make the basis for advanced control architectures for 5G networks, we derive a theoretical analysis. Specifically, we $(i)$ present a theoretical study on the limitations of eICIC and on the advantages of mmD2D sidelinks, $(ii)$ derive stochastic conditions to show how eICIC can be orchestrated to steer user fairness, $(iii)$ formulate novel and convex optimization problems to set stochastic eICIC activity patterns by leveraging the advantages of mmD2D relay and $(iv)$ show that the joint operation of stochastic eICIC and mmD2D is practical and brings dramatic gains with respect to state-of-the-art solutions. 
%\vm{manca una frase per quantificare il gain che si pu\`o ottenere}

The rest of the article is structured as follows. 
%\vm{To be reviewed}
Section~\ref{s:overview} presents the cellular framework we focus on. 
In Section~\ref{s:analysis} we derive a novel analytical model to study the network behavior. In Section~\ref{s:example} we exemplify the impact of eICIC and mmD2D on a realistic network topology. In Section~\ref{s:optimization} we formulate two new problems for the stochastic optimization of eICIC in presence of mmD2D sidelinks used for relay, respectively under static and dynamic user density conditions. Section~\ref{s:implementation} presents a possible implementation for our solution,
while in  Section~\ref{s:eval} we validate the model and report on performance evaluation. 
Section~\ref{s:related} discusses the related work and Section~\ref{s:conclusions} summarizes and concludes the article.

\section{A Novel D2D-assisted ICIC Framework}
\label{s:overview}

%{\color{red} ROAD MAP , FRAMEWORK DESCRIPTION }

We consider downlink transmissions in a cellular access network with a set $\mathcal{B}$ of interfering gNBs, operated on the same frequency band by the same operator. 
Users are provided with multi-RAT connectivity, i.e., LTE-A and IEEE 802.11ad physical interfaces. Base stations implement a subframe muting technique to control interference (ABS), while users can leverage D2D sidelinks to form relay groups. The base station elects a relay within the group. The relay is the responsible to handle the traffic of the entire group. The set of groups will be denoted by $\mathcal{C}$, and the size of group $c \in \mathcal{C}$ will be denoted as $U_c=|\mathcal{C}|$. 
Intra-group relay transmissions adopt WiGig, and the base stations select relay nodes opportunistically, according to whom is experiencing the best channel condition, similarly to what implemented with WiFi-Direct in~\cite{AMG16}. Throughout the paper we use {\it mmWave outband D2D} (mmD2D) to refer to intra-group D2D relay.  
We assume that all groups always have packets to receive, i.e., the downlink queue of each users' group is saturated.

In such mmWave D2D-assisted cellular framework, depicted in Fig.~\ref{fig:framework}, we propose a solution that retains the key strengths provided by mmD2D and ABS. However, differently from standard applications, we design a practical scheme to tune the use of ABS {\it stochastically}, ensuring user fairness additionally to inter-cell interference reduction, while at the same 
time counting on mmD2D to boost the system throughput by means of packet relay.
The building blocks of the framework outlined above are \emph{user groups}, \emph{mmWaves D2D sidelinks}, \emph{opportunistic relay election}, and \emph{ABS pattern generation}.

\ifsingle
\begin{figure} [!t]
\centering
	\includegraphics[trim={0 2cm 0.5cm 0},clip,width=0.65\textwidth]{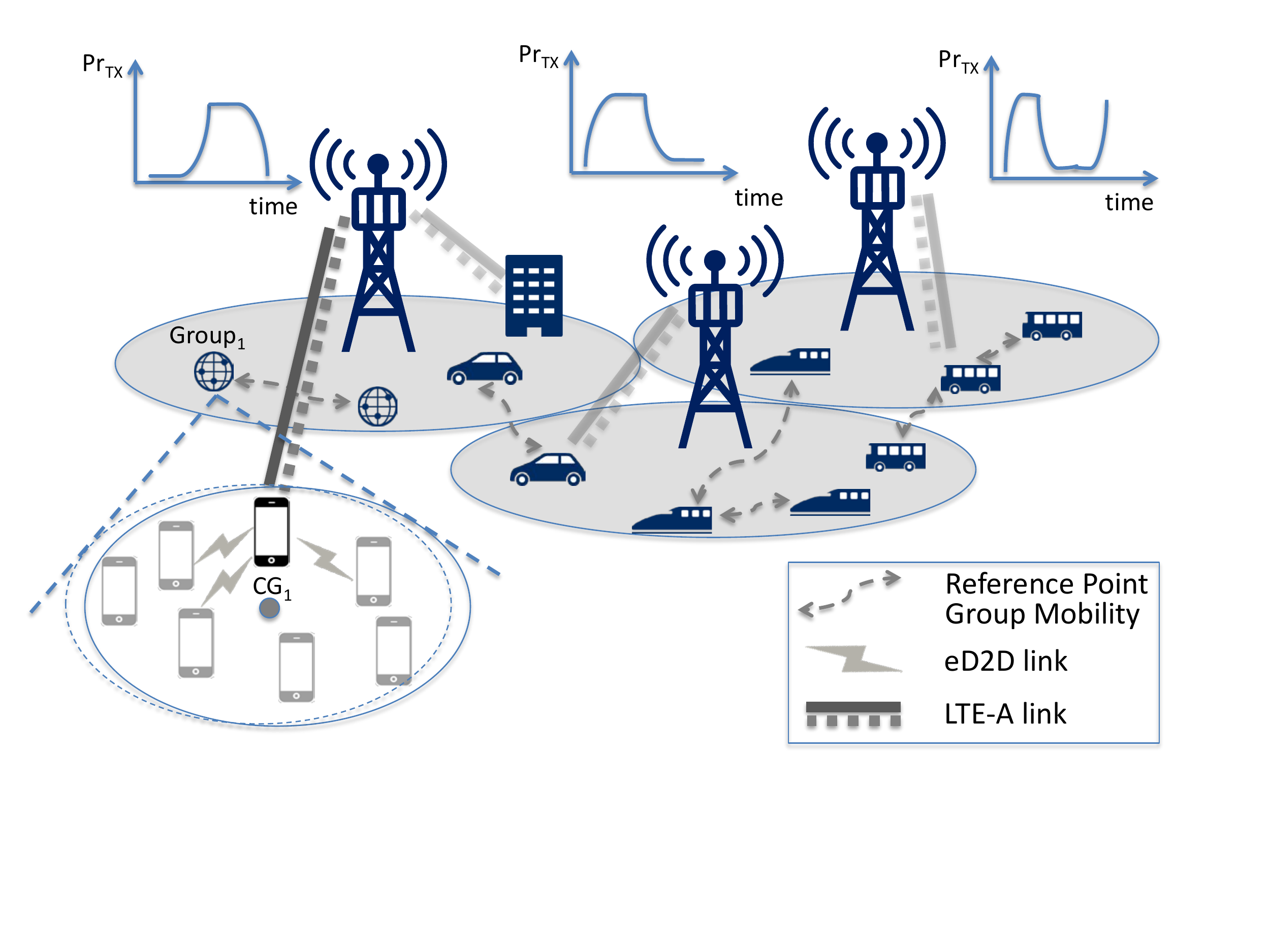}
		\vspace{-2mm}
	\caption{\vs{TO BE CHANGED} Enhanced D2D assisted cellular framework.}
	\label{fig:framework}
%	\vspace{-1em}
\end{figure}
\else
\begin{figure} [!t]
\centering
	\includegraphics[trim={0 2cm 0.5cm 0},clip,width=0.45\textwidth]{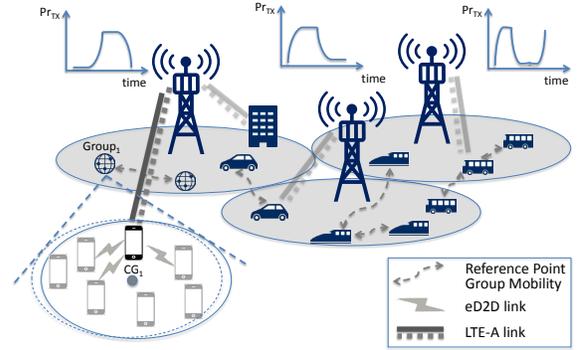}
		\vspace{-2mm}
	\caption{%\vs{TO BE CHANGED} 
	Enhanced D2D assisted cellular framework.}
	\label{fig:framework}
%	\vspace{-1em}
\end{figure}
\fi

We go beyond existing schemes, and apply the ABS paradigm to mmD2D-enabled networks wherein the use of bad channels is limited and the number of users accounting for interference reduces to the number of (mmD2D) groups. 
With our proposal, we neither impose the burden of making fully tailored per-cell ABS decisions in a centralized way, nor force any base station to generate myopic ABS patterns completely unaware of others' cell activity. Instead, we let a {\it network controller} generate and distribute stochastically built ABS patterns at low complexity.
% by 
%reducing the number of cellular users to a selected set of D2D relay nodes chosen because of their temporary good channel conditions. 
%\st{How often a relay node will receive data is a consequence of which ABS pattern has been selected by a  network controller, which stochastically builds the patterns to be distributed to the base stations.} 
Interestingly, this can be adjusted over time to pursue not only system throughput but also fairness. 
As analytically shown in Sections~\ref{s:analysis} to \ref{s:optimization}, a {\it stochastic} throughput-efficient and fairness-optimal strategy for ABS can be derived by leveraging the impact of opportunistic relay of cellular data traffic on the transmission efficiency achieved under (chosen) ABS blanking patterns.

\section{Analysis}
\label{s:analysis}

With the framework described in the previous section, we analyze the system to characterize transmission efficiency and throughput experienced by the users for each possible combination of transmitting gNBs. The analysis provides us with the tools needed to optimize the generation of gNB activity patterns stochastically. In addition, the analysis covers the case wherein mmD2D can be disabled so as to eventually evaluate analytically the overall impact. 

We compute the throughput of a user (or an mmD2D group of users) in isolation, i.e., when it attains all available cellular resources. We then model the system throughput in presence of multiple users/groups, i.e., the overall volume of traffic served by the RAN, assuming that user positions in the group are known. We finally show system throughput variations when movements of the relay groups occur. The overall model provides a useful guideline on the system performance optimization. 

%\subsection{Cooperative Group Analisys}
\subsection{Transmission Efficiency}
\label{transmission_efficiency}
%We consider a multi-cell scenario, in which a set  $\mathcal{B}$ of base stations are placed. 
%with a fixed inter-site distance (ISD). Please note that the analysis provided can be easily extended to cases in which base stations are not regularly placed in the scenario.
% \addcomment{for instance, London city deployment.}
Interfering cellular transmissions may cause a severe performance degradation. 
To analytically quantify the impact of interference we might use the Signal-to-Interference-plus-Noise-Ratio (SINR). Nevertheless, SINR does not account for the specific MCS used, and therefore, it does not account for the real efficiency of a transmission. We formally define the transmission efficiency of a group (of one or more users) as follows:

% \begin{definition}[User's Transmission efficiency]
% The efficiency $\zeta(x,y)$ of base station transmissions operated towards a destination located at $(x,y)$ is the average number of bits transmitted by the base station to that destination for each transmission symbol, considering that the base station always uses the fastest MCS allowed by the experienced SINR.
% \end{definition}

%\begin{definition}[mmD2D group's transmission efficiency]
\begin{definition}[Transmission efficiency]
The efficiency $\zeta_c(X_c,Y_c)$ of base station transmissions operated towards a group $c$ whose $U_c$ members are located at $(X_c,Y_c)=[(x_{1},y_{1}),...,(x_{U_c},y_{U_c})]$ is the average number of bits transmitted by the base station to the group for each transmission symbol, considering that the base station always uses the fastest MCS allowed by the best SINR experienced by the users in the group.
\end{definition}

%The throughput of a user or mmD2D group that obtains all transmission resources of a base station is proportional to the above-defined transmission efficiency. Thereby, modeling the transmission efficiency suffices to characterize both the maximum potential per-user throughput and the aggregate system throughput.

Computing transmission efficiency does not depend solely on the position and number of users in the group, but also on the mapping between SINR and MCS (for further details on MCS mapping examples we refer the reader to \cite{absf}). Depending on the distribution of the SINR, we may achieve largely different values for the transmission efficiency, even with the same SINR average. 
%We formalize the transmission efficiency computation in the following proposition.
Specifically, given the location of the users of group $c$ at positions $(X_c,Y_c)$ 
and denoting by $b_k$ the number of bits transmitted per symbol using MCS $k$,
the transmission efficiency for the group is
\begin{align}
\zeta_c (X_c, Y_c) & = 
\sum_{k \in \mathcal{M}} b_k 
\left [ 
F_c(T_k^{\max}) - F_c(T_k^{\min})
\right ], 
\label{eq:zetaformula}
\end{align}
where $F_c(x)$ denotes the cumulative distribution of the r.v. 
$\gamma_c (X_c,Y_c) = \max_{i \in c}\left\{\gamma_{ci} (X_c,Y_c)\right\}$,
being $\gamma_{ci} (X_c,Y_c)$ the SINR of user $i$ in group $c$. The summation in \eqref{eq:zetaformula} accounts for the number of bits per symbol transmitted by the users on a discrete set $\mathcal{M}$ of MCSs, as suggested by the standard~\cite{absf}, i.e., by casting the SINR function in a continuous subset of values comprised between $T_k^{\min}$ and $T_k^{\max}$, representing lower and upper SINR levels, respectively, for assigning the MCS $k$.

% \begin{proposition}
% \label{prop:cq_gen} 
% Given the location of the users of group $c$ at positions $(X_c,Y_c)$, the transmission efficiency for the group is computed as:
% %
% \begin{align}
% \zeta_c (X_c, Y_c) & = 
% \sum_{k \in \mathcal{M}} b_k 
% \left [ 
% F_c(T_k^{\max}) - F_c(T_k^{\min})
% \right ], 
% \label{eq:zetaformula}
% \end{align}
% %
% where $b_k$ are the bits transmitted per symbol using MCS $k$ and $F_c(x)$ is the cumulative distribution of the r.v. 
% %
% $$\gamma_c (X_c,Y_c) = \max_{i \in c}\left\{\gamma_{ci} (X_c,Y_c)\right\},$$
% %
% being $\gamma_{ci} (X_c,Y_c)$ the SINR of user $i$ in group $c$ when users are located at positions $(X_c,Y_c)$. 
% \end{proposition}

% \begin{proof}
% In our proposed scheme, the base station always transmits to the user with the highest SINR, which justifies the use of the r.v. $\gamma_c$ in Proposition~\ref{prop:cq_gen}. The summation in \eqref{eq:zetaformula} accounts for the number of bits per symbol transmitted by the users on a discrete set $\mathcal{M}$ of MCSs, as suggested by the standard~\cite{absf}, i.e., by casting the SINR function in a continuous subset of values comprised between $T_k^{\min}$ and $T_k^{\max}$, representing lower and upper SINR levels, respectively, for assigning the MCS $k$. 
% \end{proof}
% %

To compute the distribution $F_c(x)$, we consider that the SINRs experienced by users are independent random variables with averages imposed by the actual user positions. The independence comes from the fact that the fast fading process is very much affected by tiny position differences. 
As concerns the computation of the average of such random variables, we will consider two extreme cases and show that they yield similar results in practice. Specifically, we either consider that the average SINR depends on the exact position of the user, or that that position of a user can be approximated with the center of gravity of its group, so that all users see the same average signal and average interference.  

Having assumed that the SINR values in a group are independent, and considering that $\max_{i \in c}\left\{\gamma_{ci} (X_c,Y_c)\right\} \le x \iff  \forall i \in c, \gamma_{ci} (X_c,Y_c) \le x$, we can write 
\begin{align}
F_{c} (x) &= \prod_{i\in c} F_{\gamma_{ci}} (x).	
\label{eq:Fc_exact}
\end{align}

Instead, if we approximate the positions of the users with the center of gravity of the group, denoting by $\gamma_c^{\star}$ the SINR computed at the center of gravity, the following expression holds and approximates~\eqref{eq:Fc_exact}: 
\begin{align}
F_c (x) &= \left[F_{\gamma_c^{\star}}(x) \right]^{U_c}. 
\label{eq:Fc_center}
\end{align}
This approximation makes sense when a group mobility model
%, as the RPGM one mentioned in Section~\ref{ss:groups}, 
can be used to describe the dynamics of the user's topology.

Moreover, as pointed out in~\cite{Meyer73}, in a urban environment, the power received by a user from a base station at any given location follows a negative exponential distribution (whereas the instantaneous signal follows a Rayleigh distribution) whose average value only depends on the pathloss effect. 
%The noise is either negligible with respect to the interference or can be treated as a (small) constant term.  
Therefore, the distributions $F_{\gamma_{ci}}$and $F_{\gamma_{c}^{\star}}$ can be computed using the result reported in the following proposition.  

\begin{proposition}
\label{prop:CDF}
The distribution of the SINR $\gamma$ resulting from an exponential useful signal with average power $1/\lambda_S$, $k$ independently exponentially distributed interfering signals $I_j$ with average power $1/\lambda_j$,  
and additive Gaussian white noise with zero mean and power $N$ is, $\forall x \ge 0$,
% %
% \begin{align}
% F_{\gamma}(x) & = 1 - e ^ {-\lambda_S N x} \prod_{j=1}^k \frac{\lambda_j}{\lambda_j + x\lambda_S }.
% \label{eq:sinr}
% \end{align}
% %
%
\begin{align}
    \!\!\!F_{\gamma}(x) & = 1 - \frac{1}{\sqrt{1 + 2 \lambda_S N x}} \prod_{j=1}^k \frac{\lambda_j}{\lambda_j + x\lambda_S }.
    \label{eq:CDF}
\end{align}
\end{proposition}

\noindent The proof is reported in the Appendix.

The average power levels used in the proposition can be computed with a legacy pathloss model, based on distances between signal/interference sources and receivers, and environmental parameters~\cite{seybold2005introduction}.

Let us define each potential combination of active base stations as one of the possible {\it states} of the system, and let us denote it by $s$. Let $\mathcal{B}_s$ be the set of base stations transmitting data in state $s$, whereas $|\mathcal{S}|=2^{|\mathcal{B}|}$ be the set of all possible states. 
With each state $s$ the interference changes and so the transmission efficiency does, as expressed in the following definition. 

\begin{definition}[Transmission efficiency in state $s$]
We denote by $\zeta_c^s$ the transmission efficiency of group $c$ in state $s$, which can be computed with \eqref{eq:zetaformula}--\eqref{eq:CDF} by ignoring the contributions due to base stations not active in state $s$. The transmission efficiency of a group under the coverage of an inactive base station is set to $0$.
\end{definition}

% for given group coordinates $(x,y)$ as well as the average spatial transmission efficiency $\chi_\mathcal{C}|_s$. 
%\st{Of course, $\zeta^s_c=0$ when group $c$'s base station is prevented from transmitting in state $s$. }

With the transmission efficiency derived above, we are now ready to formulate the throughput in each possible system state.

\subsection{Instantaneous System Throughput}
\label{s:cell_throughput}
\vspace{-1mm}
%Here we analyze the overall network performance in presence of multiple groups sharing the available resources. 
%We consider a set of $N$ groups, each indicated as $\mathcal{C} \in \{1, \dots, N\}$ and comprising $U_\mathcal{C}$ users. 
%
%As mentioned in Section~\ref{ss:groups}, 
Each active base station implements a scheduler such that each group retrieves transmission opportunities proportionally to the group size, i.e., for any group $c \in \mathcal{C}$ under the coverage of an active base station $b \in \mathcal{B}_s$ at time $t$, the resources allocated are expressed as follows:
\vspace{-1mm}
\begin{equation}
\label{eq:offered_traffic}
D_c (t)= K_\text{sym} 
	\frac 
	{ U_c }
	{ \sum_{i  | {b \text{ covers groups } i \text{ and } c} \text{ at time } t} U_i }, 	
%	\vspace{-2mm}
\end{equation}
where $K_\text{sym}$ is the total available number of symbols per second at the base station $b$ serving group $c$ at time $t$. Please note that $D_c (t)$ is independent of the particular state $s$ of the ABS if the system is in saturation; otherwise, in \eqref{eq:offered_traffic} we should use the actual number of backlogged groups at time $t$ instead of the total number of groups.

With the above, the resulting instantaneous per-group throughput $\Gamma^s_c (t)$ and the corresponding aggregate system throughput $\Gamma^s (t)$ in state $s$ at time $t$ are computed with the following expressions:
\begin{align}
\label{eq:throughput}
\Gamma^s_c (t) & = D_c(t) \; \zeta^s_c(x_c,y_c);\\
\Gamma^s (t) & = \sum_{c \in \mathcal{C}} \Gamma_c^s (t). 
\end{align}

%The above expressions can be used to track throughput performance over time as ABS state and group locations changes.

%
%The corresponding instantaneous system throughput 
%in state $s$ at time $t$ 
%is 
%$\Gamma^s (t)= \sum_{c \in \mathcal{C}} \Gamma_c^s (t)$. 

\subsection{Asymptotic Performance}
\label{ss:asymptotic}

Let us now consider the impact of group mobility and ABS state on asymptotic performance, since they directly affect experienced SINR levels. The objective is to compute the mean transmission efficiency and throughput of group $c$, averaged over time. %\st{and space}.

\subsubsection{Asymptotic Transmission Efficiency}

The asymptotic transmission efficiency is defined as the following time average:
\begin{align}
\bar{\zeta_c} & = \lim_{T\rightarrow +\infty} \frac { \int_0^T \zeta_c^{s(t)}(x(t),y(t)) dt} { T },
\end{align}
where ABS state $s$ and group's location can change over time. Assuming that the system is ergodic, the above quantity is equivalent to the stochastic average of a random process $\zeta_c^{V}$ (note that we use $V$ to refer to a random process that represents the ABS state):
\begin{align}
\bar{\zeta_c} & = E[\zeta_c^V] = \sum_{s \in \mathcal{S}} E[\zeta_c^s] P^s,
\label{eq:avg_eff}
\end{align}
where we have used the total probability formula and defined $P^s \!=\! \text{Pr}(V=s)$. The above relation unveils that the asymptotic transmission efficiency can be expressed in terms of per-state transmission efficiency. Most importantly, it conveys that the asymptotic transmission efficiency is affected by the probability of using a particular ABS state, and the following fundamental result holds: 

\begin{proposition}
The asymptotic transmission efficiency $\bar{\zeta_c}$, for a given user mobility model and for a fixed topology of base stations is solely affected by two components that can be tuned independently: ABS state probabilities and mmD2D group composition. 
\label{pr:eff}
\end{proposition}
\begin{proof}
Note that transmission efficiency in a given state can be expressed as a function of the location. Specifically, denoting the spatial distribution of group $c$'s center of gravity by $L_c(x,y)$, and the coverage area of base station $b$ by $A_b$, we have the following simple expression for the conditional average transmission efficiency of a group in state $s$ under the coverage of base station $b$:
\begin{equation}
E[\zeta_c^s|b] = \frac { \int_{A_b} L_c(x,y) \zeta_c^s(x,y) dA } { \int_{A_b} L_{c}(x,y)dA },
\label{eq:zcb}
\end{equation}
The denominator of the RHS in \eqref{eq:zcb} is the probability $p_c(b)$ that group $c$ is under the coverage of base station $b$.  
%
% is also a function of $L_c(x,y)$:
% %
% \begin{align}
% p_c(b) &= \text{Pr}(\text{$c$ covered by $b$}) = \int_{A_b} L_{c}(x,y)dA.
% \label{eq:pcb}
% \end{align}
%
Therefore, applying the total probability formula yields: 
%
% \begin{align}
%  E[\zeta_c^s] & = \sum_{b \in \mathcal{B}} E[\zeta_c^s | b ] \; p_c(b) \nonumber \\ 
%  & = \sum_{b \in \mathcal{B}} \frac { \int_{A_b} L_c(x,y) \zeta_c^s(x,y) dA } { \int_{A_b} L_{c}(x,y)dA} \cdot \int_{A_b} L_{c}(x,y)dA \nonumber\\
%  & = \int_A L_c(x,y) \zeta_c^s(x,y) dA,
% \label{eq:state_eff} 
% \end{align}
% 
\begin{align}
 E[\zeta_c^s] & = \sum_{b \in \mathcal{B}} E[\zeta_c^s | b ] \; p_c(b) = \int_A L_c(x,y) \zeta_c^s(x,y) dA,
\label{eq:state_eff} 
\end{align}
where $A \! = \! \cup_{b \in \mathcal{B}} {A_b}$ is the total covered area.
%and the conditional transmission efficiency of a group depends on its composition (besides mobility parameters). 

Using \eqref{eq:state_eff} in \eqref{eq:avg_eff} shows that the asymptotic transmission efficiency, for a given base station topology and mobility model, can be tuned by means of two independent mechanisms, namely $(i)$ adjusting ABS state probabilities $P^s$ and $(ii)$ using different group sizes, which affects $\zeta_c^s$.
\end{proof}

\begin{corollary}
When the position of a group is known in terms of spatial distribution $L_c(x,y)$, its asymptotic transmission efficiency can be computed as
\begin{align}
\bar{\zeta_c} & = \sum_{s \in \mathcal{S}} P^s \;  \int_A L_c(x,y) \zeta_c^s(x,y) dA.
\end{align}
\end{corollary}

%, when a given state $s$ of the ABS is considered. %As previously said, the mobility of each group can be approximated by the mobility of its gravity center. 
%We rely on the knowledge of the asymptotic distribution of group locations and we indicate the spatial distribution of group $c$'s center of gravity over the scenario as $L_c(x,y)$, considering that the movements of the overall group can be approximated by the mobility of its center of gravity. Such distribution can be easily derived for well-known mobility models, e.g., the Random Way Point (RWP) mobility model \cite{waypoint} that applies to the gravity center's position if groups follow RPGM, and it can also be empirically retrieved if needed.
%
%Based on the spatial distribution of the group gravity center, and fixing the ABS state to a particular state $s$, the {\it conditional} average transmission efficiency achieved by a group under a particular base station $b$, namely $E[\zeta_c^s|b]$, can be computed as follows:
%
%\begin{equation}
%E[\zeta_c^s|b] = \frac { \int_{A_b} L_c(x,y) \zeta_c^s(x,y) dA } { \int_{A_b} L_{c}(x,y)dA },
%\end{equation}
%where $A_b$ is the area covered by base station $b$. 

\subsubsection{Asymptotic Average Throughput}

We can derive the asymptotic average throughput (over time) achieved by group $c$ with an approach similar to the one described above for the asymptotic transmission efficiency:
\begin{align}
\label{eq:gammac}
\bar{\Gamma_c} & 
= E[\Gamma_c]
= \sum_{s \in \mathcal{S}} E[\Gamma_c^s] P^s.
%= \sum_{s \in \mathcal{S}} P^s \int_A L_c(x,y) \Gamma_c^s(x,y) dA.
\end{align}
In the above formula, the conditional average throughput
in state $s$ can be computed as follows: 
\begin{align}
E[\Gamma_c^s] & = \sum_{b\in\mathcal{B}_s} p_c(b) \, E[\zeta_c^s|b] \, E[D_c|b] \nonumber \\
& = \sum_{b\in\mathcal{B}_s} E[D_c|b] \int_{A_b} L_c(x,y) \zeta_c^s(x,y) dA,
\label{eq:athr}
\end{align}
where $E[D_c|b]$ represents the average number of symbols allocated to group $c$ under base station $b$.
%, while $p_c(b)$ and $E[\zeta_c^s|b]$ have been defined in the proof of Proposition~\ref{pr:eff} in \eqref{eq:zcb} and \eqref{eq:pcb}, respectively.
%
%In reality, group $c$ shares the cellular resources with a number of other groups under the same base station. Therefore, we compute the average throughput $\bar{\Gamma}_c^s$ of group $c$ when the set of groups in the scenario is fixed as well as the state of the ABS is set to $s$. Denoting as $\mathcal{C}$ the set of groups in the scenario and as $L_i(x,y)$ the spatial distribution of each group $i\in \mathcal{C}$, 
%In turn, the average number of symbols per second to be used in \eqref{eq:athr}, i.e., 
$E[D_c|b]$, is obtained by considering all possible combinations of groups that fall under the coverage of base stations $b$, as follows:

\begin{equation}
E[D_c|b] =\!\!\!\! \sum_{Z\in \mathcal{P}(\mathcal{C},c)} \prod_{i\in Z} p_{i}(b) \prod_{j\notin Z} [1-p_{j}(b)] \frac{U_c}{U_c+\sum_{i\in Z} U_{i}} K_{sym},
\label{eq:avgDc}
\end{equation}
where $\mathcal{P}(\mathcal{C},c)$ is the power set of the groups in $\mathcal{C}$ when group $c$ is taken out ($Z$ is therefore a set too). The calculation of $E[D_c|b]$ assumes a proportional resource scheduling based upon group sizes, i.e., proportional to the number of users building up the mmD2D groups.
%\st{Interestingly, $E[D_c|b]$ can be also computed as the time average of resources received from group $c$ while scheduled by base station $b$, i.e.:}

%\begin{equation}
%\xcancel{E[D_c|b] = \lim_{t \to + \infty} \frac{\int_0^t D_c(t) \delta_c^b(t) ~dt}{\int_0^t \delta_c^b(t)~dt},}
%\end{equation}
%\st{where $\delta_c^b(t) = 1$ only if the position of group $c$ at time $t$ is served by base station $b$.}
%Given $E D}_c|b$ we can calculate $\bar{\Gamma}_c^s$ as the following:

%\begin{equation}
%\bar{\Gamma}_c^s = \sum_{b\in\mathcal{B}_s} p_c|b \cdot\bar{\zeta}_c^s|b \cdot \bar{D}_c|b.
%\end{equation}

Since \eqref{eq:avgDc} does not depend on ABS state probabilities, \eqref{eq:athr} behaves likewise, and the asymptotic average throughput \eqref{eq:gammac} has a similar structure as the transmission efficiency. Therefore, the above derivation directly leads to the following result, similar to what found for the transmission efficiency:
\begin{proposition}
The asymptotic throughput $\bar{\Gamma_c}$, for a given user mobility model and for a fixed topology of base stations is solely affected by two components that can be tuned independently: ABS state probabilities and mmD2D group composition.
\label{pr:thr}
\end{proposition}

The distribution of resources expressed in \eqref{eq:avgDc}---and therefore the asymptotic system throughput in state $s$ expressed in \eqref{eq:athr}---is strictly dependent on the set of active groups $\mathcal{C}$ and their movements. 
%In case $\mathcal{C} = c$, $\bar{\Gamma}_c^s$ is represented as in Eq.~\eqref{eq:thr_lonely}. Otherwise dealing with 
In case of homogeneous scenarios wherein all groups experience the same spatial distribution of their gravity center, i.e., $L_i(x,y)=L_c(x,y)$, so that $p_i(b) = p_c(b)$, $\forall i \in\mathcal{C}, b \in \mathcal{B}$, and all groups show the same number of users, i.e., $U_i=U_c$, $\forall i \in\mathcal{C}$, $E[D_c|b]$ is the same for all groups and can be simplified as follows:

\begin{equation}
E[D_c|b] = \sum_{k = 0}^{|\mathcal{C}|-1} {{|\mathcal{C}|-1}\choose{k}} \frac{K_{sym}}{k+1} [p_c(b)]^k [1-p_c(b)]^{|\mathcal{C}|-k-1}.
\end{equation}

The asymptotic system throughput achieved under ABS state $s$ is simply given by the sum of group's asymptotic throughputs, i.e., $E[\Gamma^s] = \sum_{c \in \mathcal{C}} E[\Gamma_c^s]$. 
%\st{The same result would be achieved for the asymptotic  system throughput by directly integrating over time the sum of $\Gamma^s_c (t)$ values given by (4). However, as ABS states change over time, the actual total }

\begin{remark}
Note that the case of independent users not joining any group can be analyzed by regarding those users as groups of size one.
\end{remark}

\subsection{Summary of Analytic Results}

Transmission efficiency and throughput achieved by each group---and hence the fairness level experienced in the network---depend on the fraction of time spent in each ABS state and on the gain attained by means of mmD2D when relay groups form. As we have shown through this section for instantaneous and asymptotic performance figures, the two components operate orthogonally since group composition does not depend on ABS state and vice versa. Besides ABS states and groups composition, mobility plays an important role, and we have shown how to analytically derive the achievable throughput as a function of group position distributions. 

Most importantly, all derived results show that there is no dependency on how ABS state alternate over time, except the probability to observe a state has to be known to be able to predict the performance of the system. As such, ABS state probabilities can be tuned to optimize system performance.  

Indeed, in Section~\ref{s:optimization}, we will study the optimization of %\st{such fractions of time, i.e., the optimization of} 
ABS in terms of ABS state probabilities $P^s$ enforced by means of ABS patterns computed with no need to coordinate between gNBs and with simple random assignments---thus achieving extremely low complexity---for fixed mmD2D configurations and mobility parameters. Before that, we proceed by providing the reader with a concrete example that helps quantify how mmD2D and the choice of ABS state impact on throughput and fairness performance.

\section{Example of Impact of ABS and mmD2D Relay}
\label{s:example}

To evaluate the impact of ABS on a mmD2D-enabled network, we use an example based on a realistic heterogeneous dense-urban area of $400$m~$\times$~$320$m close to the Oxford Circus metro station in London city (UK). %\st{Based on its real topography, in the example presented in what follows, we} 
We consider only the base stations under the control of O2 mobile network operator, whose base stations' position and transmission power are publicly available.\footnote{All information are retrieved from the OFCOM reports available at \url{http://stakeholders.ofcom.org.uk/sitefinder/sitefinder-dataset/}
}

\ifsingle
\begin{figure} [!t]
\centering
	\includegraphics[width=0.65\textwidth]{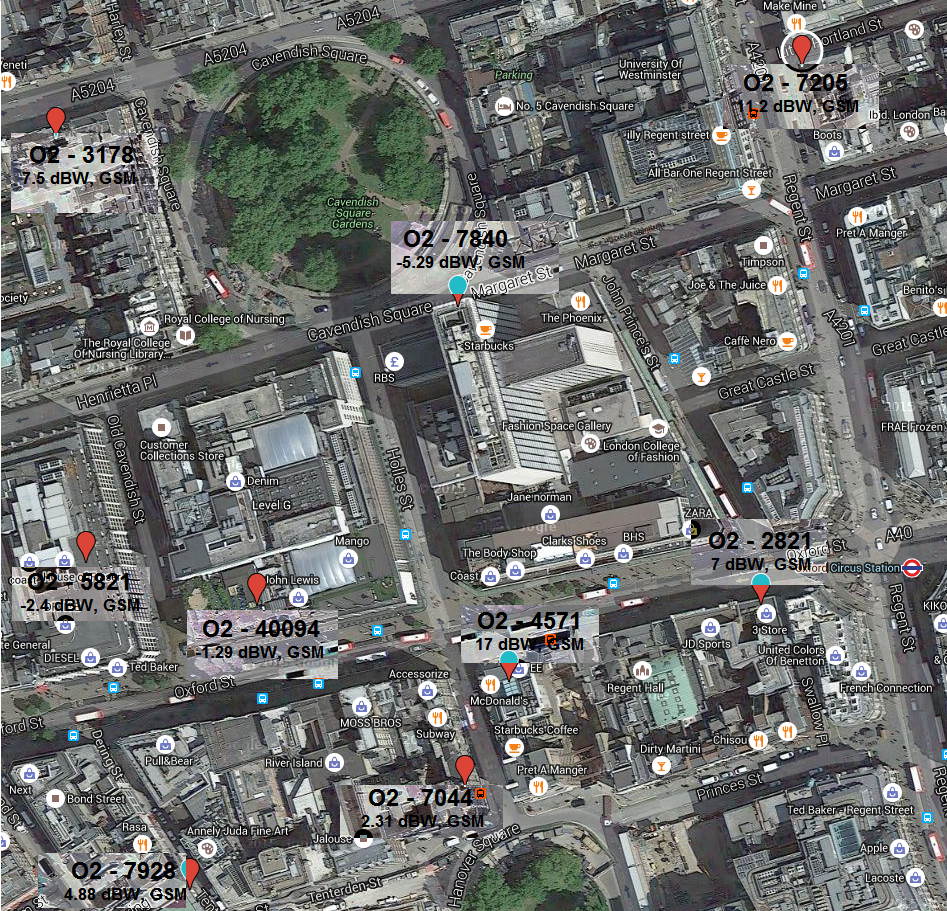}
	%\includegraphics[width=0.27\textwidth]{figs/london_picture_white.png}
%	\vspace{-3mm}
	\caption{\footnotesize O2 deployment in London city - Oxford Circus.}
	\label{fig:london_deployment}
%	\vspace{-4.5mm}
\end{figure}%
\else
\begin{figure} [!t]
\centering
	\includegraphics[width=0.3\textwidth]{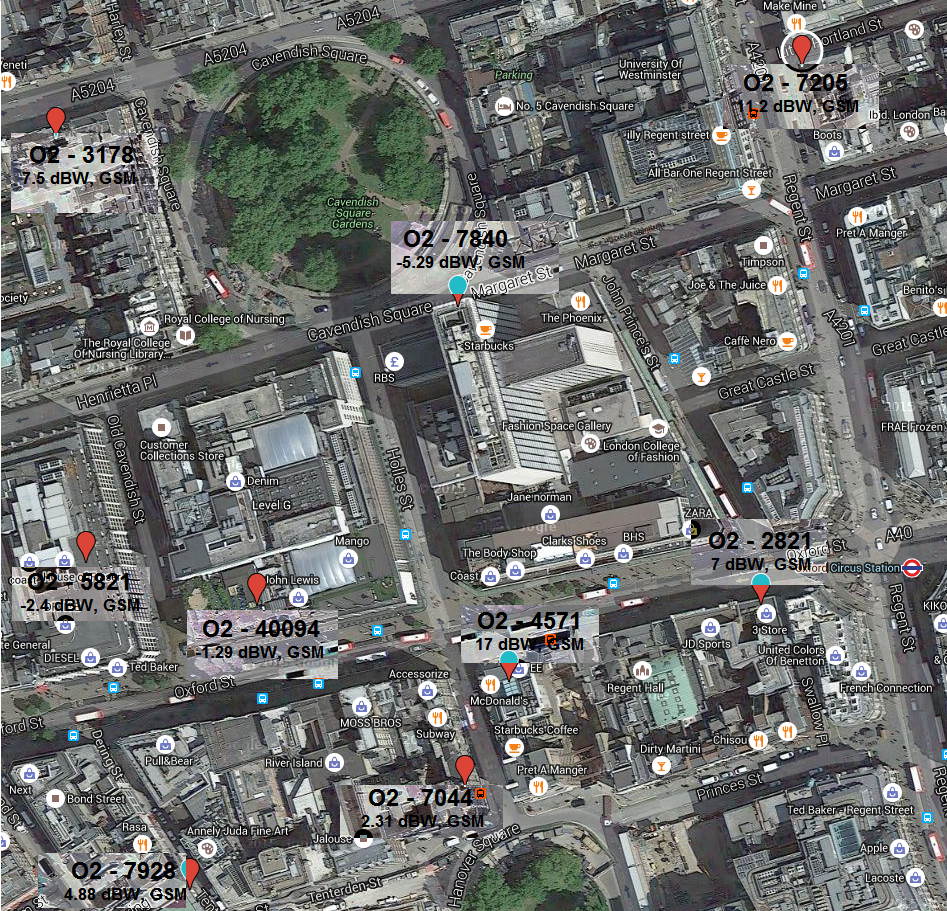}
	%\includegraphics[width=0.27\textwidth]{figs/london_picture_white.png}
%	\vspace{-3mm}
	\caption{\footnotesize O2 deployment in London city - Oxford Circus.}
	\label{fig:london_deployment}
%	\vspace{-4.5mm}
\end{figure}
\fi

%\st{In the considered area, $9$ base stations are present, as} 
As illustrated in Fig.~\ref{fig:london_deployment}, there are $9$ base stations irregularly spaced and with heterogeneous transmission powers. %\st{which also reports the heterogeneous transmission powers of the base stations. Accordingly, we plot in} 
Fig.~\ref{fig:heatmaps} reports a snapshot of the group throughputs 
%transmission efficiency $\zeta^s_c(x,y)$ 
obtained in three different ABS states for 
%for a network area of $140 \times 140$ m$^2$ with $7$ base stations equally spaced apart $50$m and using $20$ dBm as transmission power. 
a static allocation of $200$ groups, whose centers of gravity are reported as white dots in the figure. 
Fig.~\ref{fig:heatmaps} shows the average of the group throughputs over the entire network (top of each subfigure) and over the area of each cell (indicated next to the cell center) when group sizes are $U_c = 1$, i.e., without mmD2D relay, and $U_c = 5$ for all groups, respectively. 

As shown in Figs.~\ref{fig:heatmap_all} and \ref{fig:heatmap_all5}, with all base stations active (i.e., as in the case no ABS was enforced) the distribution of group throughput is unfair and low in a realistic deployment, with especially large cells guaranteeing poor throughput. The effect of relay via mmD2D sidelinks is notable and results in a $65\%$ increase in the average of group throughput over the entire network, though not all cells experience the same degree of benefit.  

Figs.~\ref{fig:heatmap_corners} and \ref{fig:heatmap_corners5} show throughput achieved by keeping active the three strongest base stations only. In this ABS state, most of the users are under coverage of active base stations, but the interference level remains high. Only the central base station experiences a significant gain (it doubles the throughput for the groups it serves, with and without mmD2D relay), since many interfering small cells around it are blanked. Here, the impact of mmD2D is less important but still high ($+51\%$ on the group throughput averaged over the entire network). Moreover, group throughput in this state is not only unfairly distributed, but also much lower than in the case with all base stations active. In our experiments, we have observed a similar behavior in many cases, which questions the ability of ABS to improve throughput. 

\ifsingle
\begin{figure} [t!]
\centering
	\subfigure[All base stations active, 
	without mmD2D ($U_{\mathcal{C}}=1$).]
	{
	\includegraphics[scale=0.30]{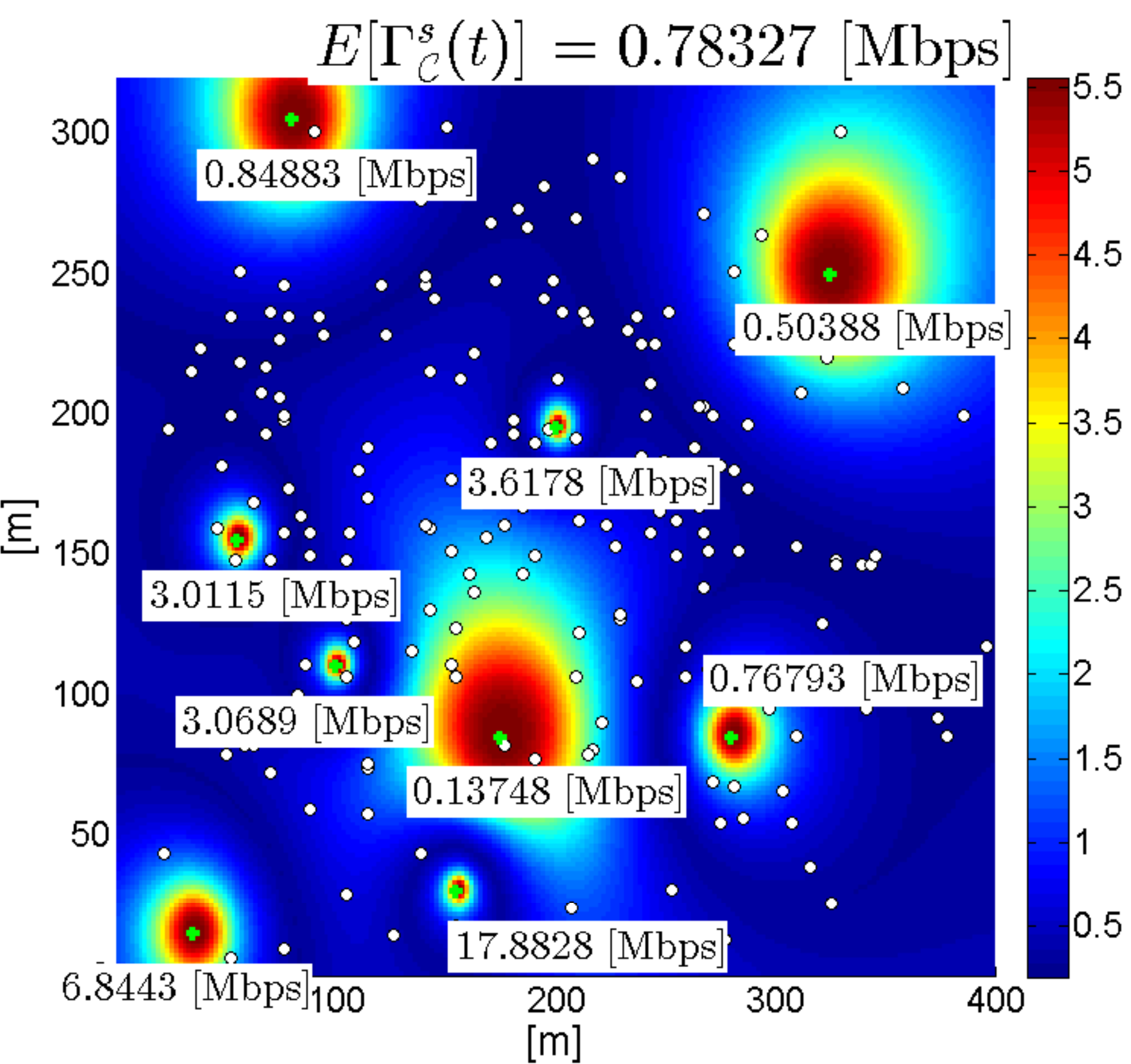}	
	\label{fig:heatmap_all}
	}
	\hspace{2cm}
	\subfigure[All base stations active, $U_{\mathcal{C}}=5$.]
	{
	\includegraphics[scale=0.30]{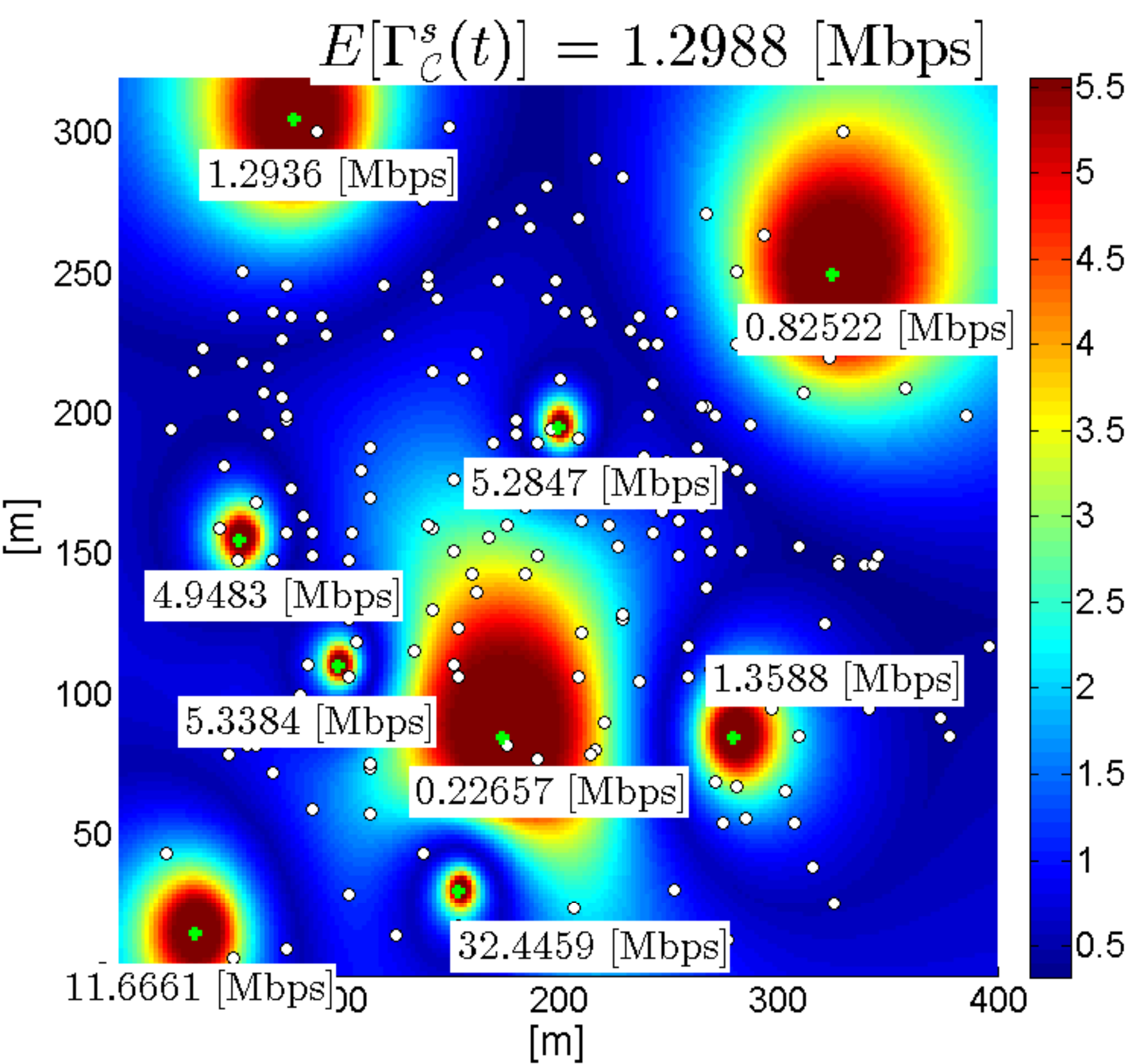}
	\label{fig:heatmap_all5}
	}
	\subfigure[Strongest base stations only, without mmD2D ($U_{\mathcal{C}}=1$).]
	{
	\includegraphics[scale=0.30]{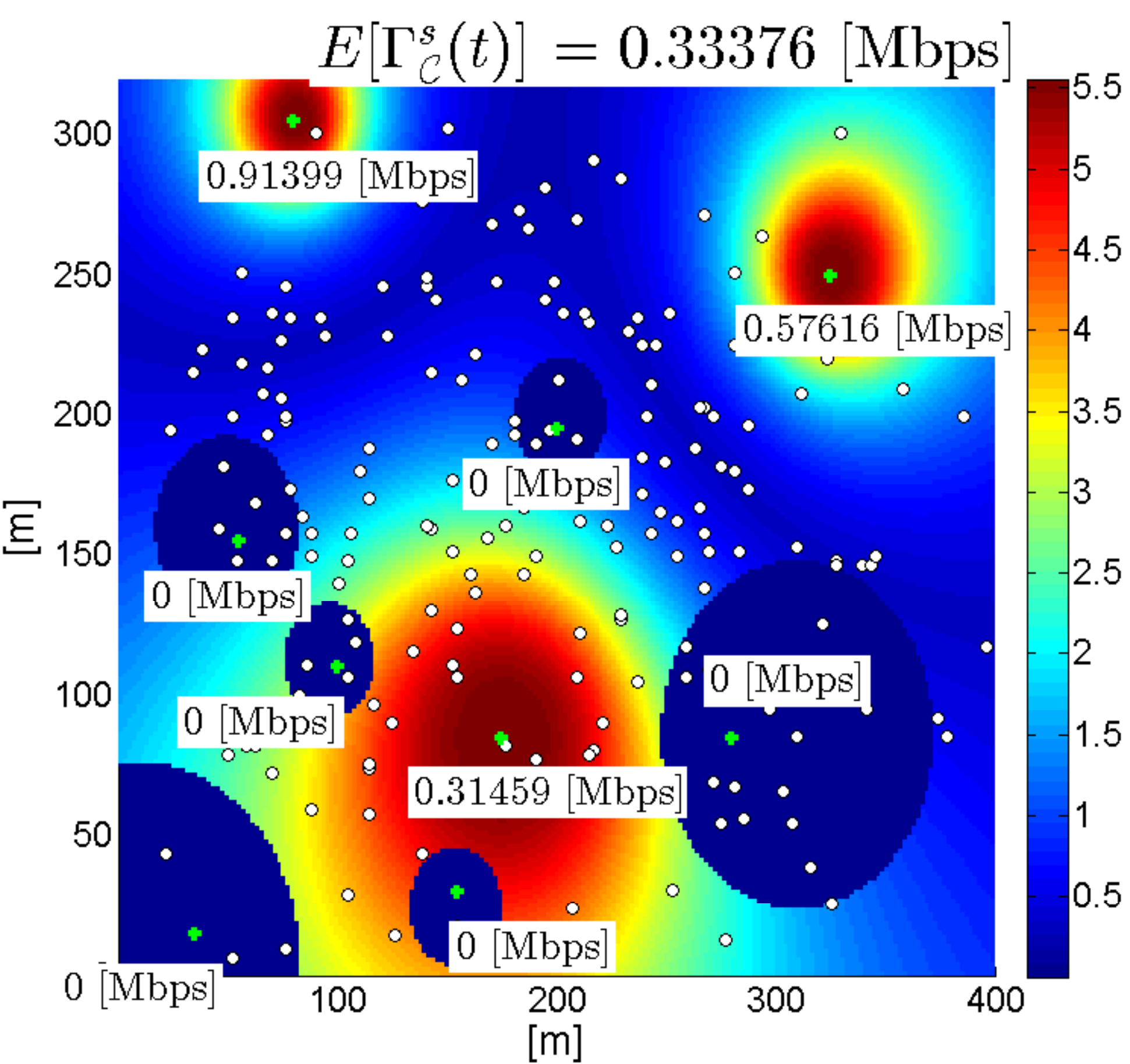}
	\label{fig:heatmap_corners}
	}
	\hspace{2cm}
	\subfigure[Strongest base stations only, $U_{\mathcal{C}}=5$.]
	{
	\includegraphics[scale=0.30]{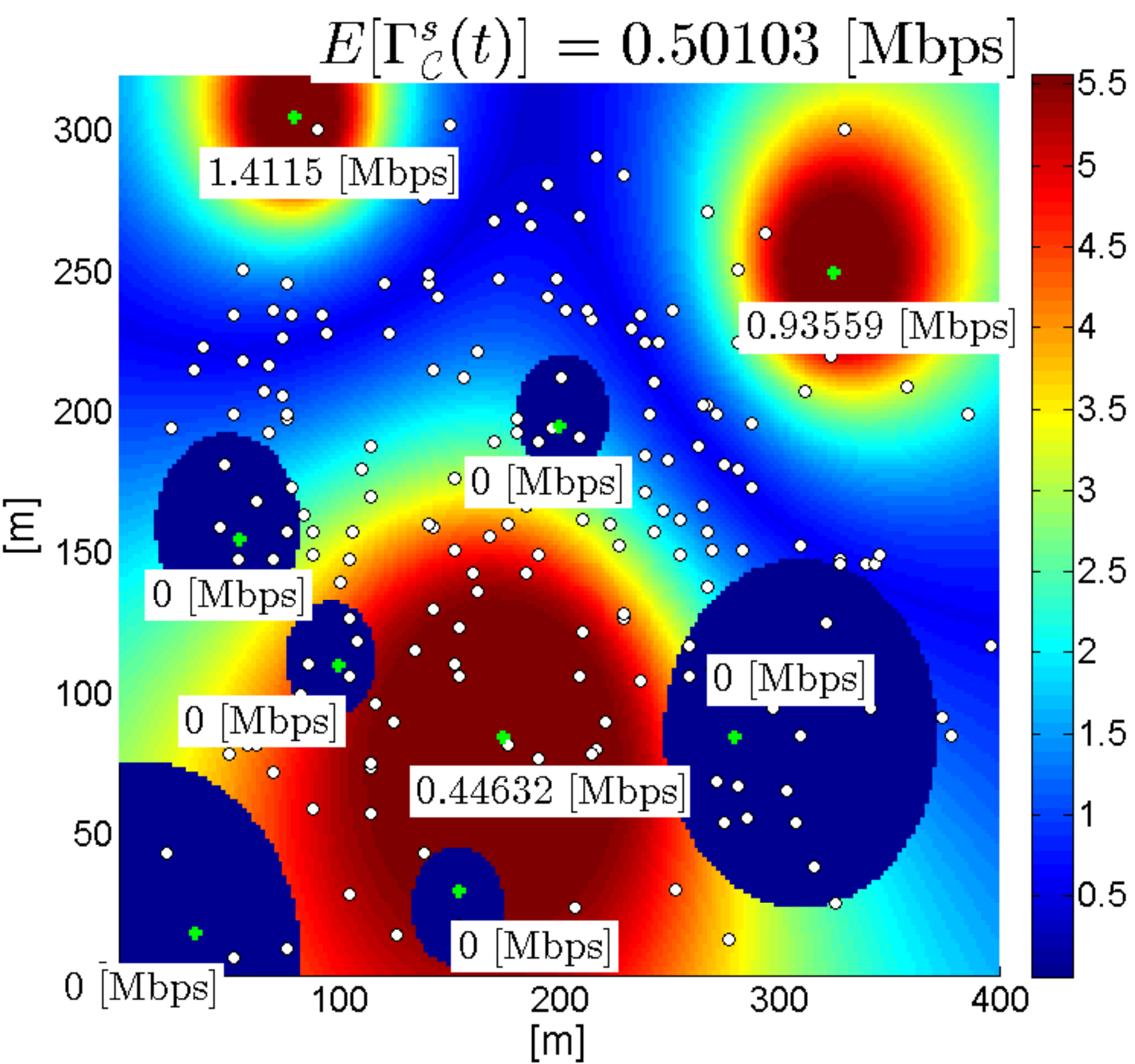}
	\label{fig:heatmap_corners5}
	}
	\subfigure[Blanking the central base
	station, without mmD2D ($U_{\mathcal{C}}=1$).]
	{
	\includegraphics[scale=0.30]{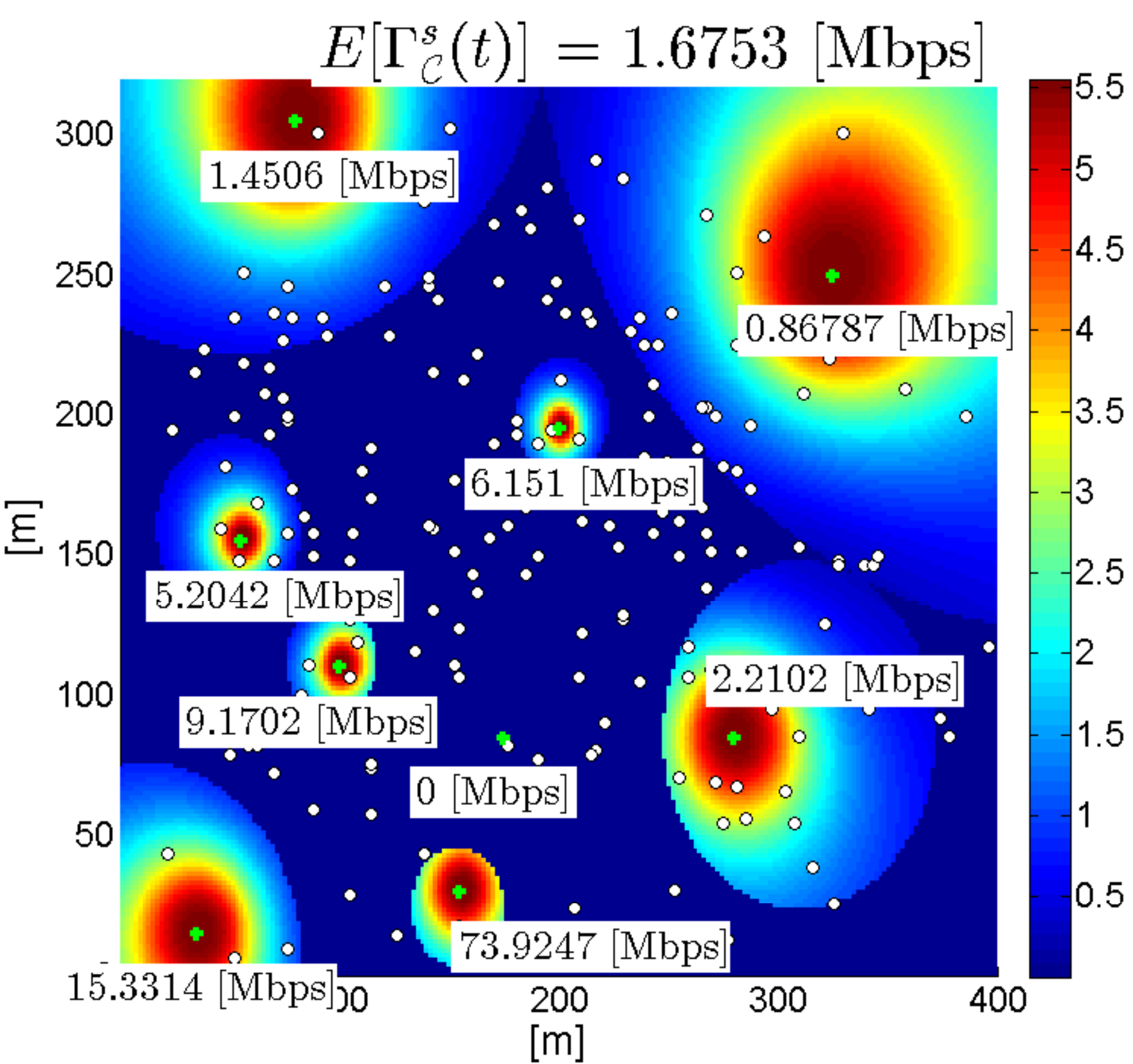}
	\label{fig:heatmap_center}
	}
	\hspace{2cm}
	\subfigure[Blanking the central base
	station, $U_{\mathcal{C}}=5$.]
	{
	\includegraphics[scale=0.30]{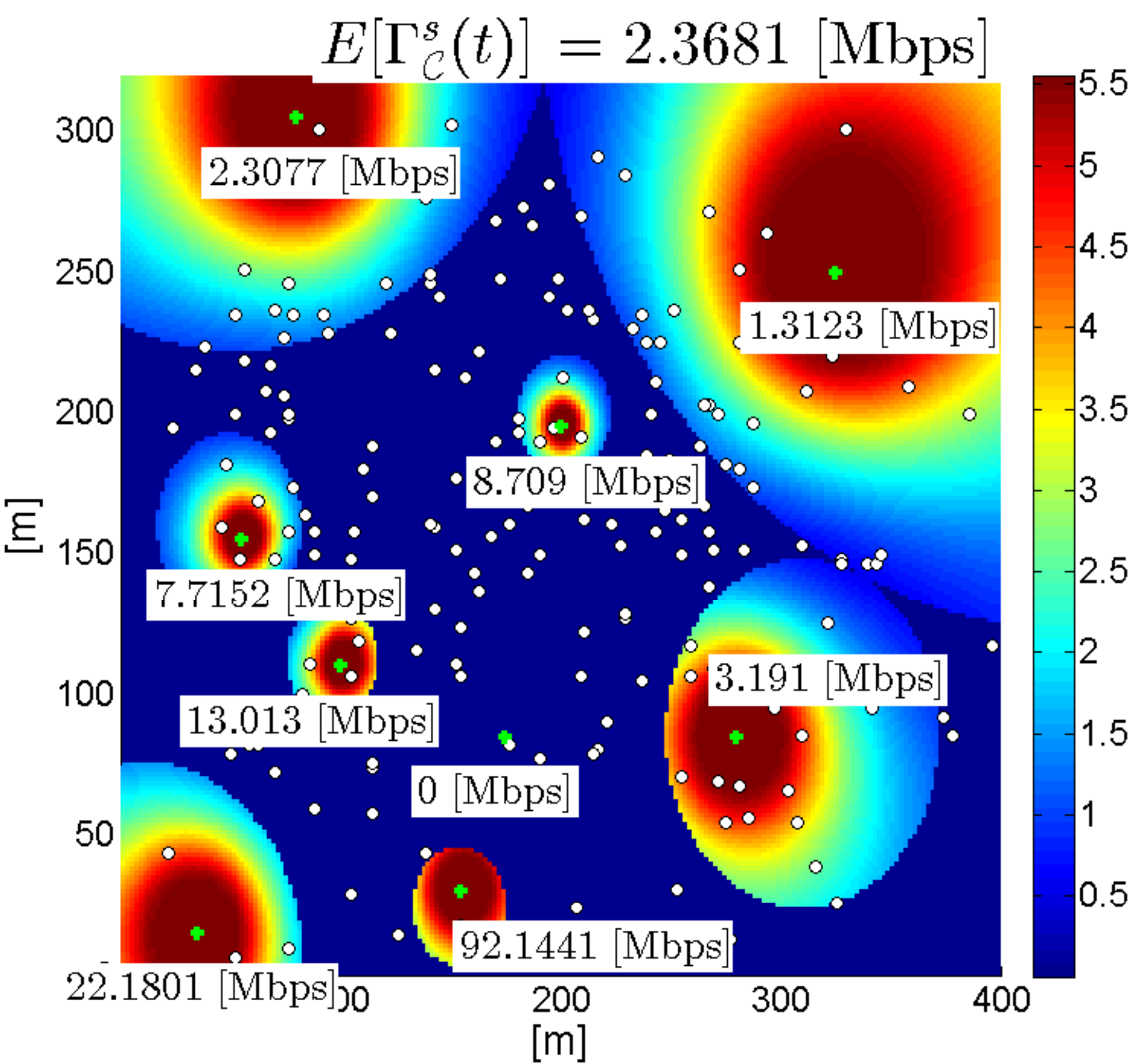}
	\label{fig:heatmap_center5}
	}
	\caption{\footnotesize
	Example of throughputs achievable in a realistic network deployment. Figure best viewed in colors.}
	\label{fig:heatmaps}
\end{figure}
\else
\begin{figure} [t!]
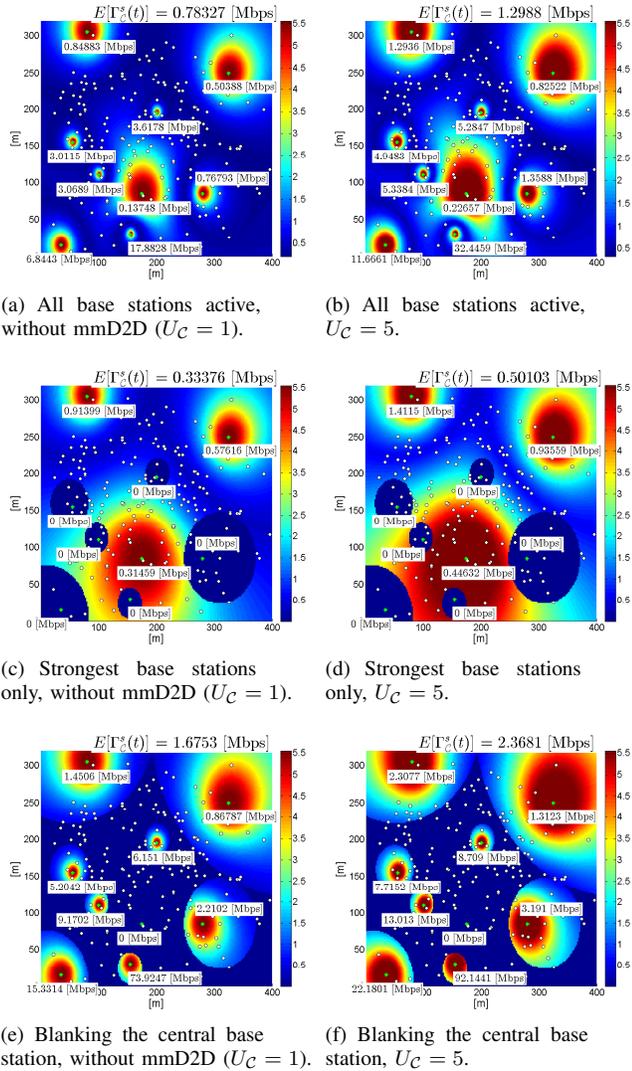

\centering
	\subfigure[All base stations active, $\qquad$ 
	without mmD2D ($U_{\mathcal{C}}=1$).]
	{
	\includegraphics[scale=0.20]{./figs/fig_2_london_all}	
	\label{fig:heatmap_all}
	}
	\subfigure[All base stations active, $\qquad$ $U_{\mathcal{C}}=5$.]
	{
	\includegraphics[scale=0.20]{./figs/fig_2_london_all_5}
	\label{fig:heatmap_all5}
	}
	\subfigure[Strongest base stations $\qquad$ only, without mmD2D ($U_{\mathcal{C}}=1$).]
	{
	\includegraphics[scale=0.20]{./figs/fig_2_london_onlymacro}
	\label{fig:heatmap_corners}
	}
	\subfigure[Strongest base stations $\qquad$ only, $U_{\mathcal{C}}=5$.]
	{
	\includegraphics[scale=0.20]{./figs/fig_2_london_onlymacro_5}
	\label{fig:heatmap_corners5}
	}
	\subfigure[Blanking the central base $\qquad$
	station, without mmD2D ($U_{\mathcal{C}}=1$).]
	{
	\includegraphics[scale=0.20]{./figs/fig_2_london_exceptcentral}
	\label{fig:heatmap_center}
	}
	\subfigure[Blanking the central base $\qquad$
	station, $U_{\mathcal{C}}=5$.]
	{
	\includegraphics[scale=0.20]{./figs/fig_2_london_exceptcentral_5}
	\label{fig:heatmap_center5}
	}
%vspace{-2mm}
	\caption{\footnotesize
	Example of throughputs achievable in a realistic network deployment. Figure best viewed in colors.}
	\label{fig:heatmaps}
%\vspace{-1.5em}
\end{figure}
\fi

Finally, Figs.~\ref{fig:heatmap_center} and \ref{fig:heatmap_center5} illustrate how blanking the central (and strongest) cell yields much higher throughput than in the other considered cases. Therefore, this state is convenient to boost throughput. However, once again, one should notice that group throughput is unfairly distributed, with and without mmD2D, and that mmD2D with groups of $5$ users brings a significant $41\%$ gain on the group throughput averaged over the entire network. 

In general, the distribution of throughput over groups is very much unfair in any ABS state, which justifies the effort to introduce a mechanism to enforce fairness.
%
%scheduling only one base station considerably improves the transmission efficiency of the group at the covered locations. However, the gain is not proportional to the number of blanked cells, as evident from a comparison with the behavior in other states (cf. Figs.~\ref{fig:heatmap_corners}, ~\ref{fig:heatmap_corners5}, ~\ref{fig:heatmap_all} and~\ref{fig:heatmap_all5}).
%
%Fig.~\ref{fig:heatmap_corners} and Fig.~\ref{fig:heatmap_corners5} represent the case in which three outer cells are scheduled by the ABS application. Such state exhibits very high transmission efficiency within the coverage of the active base stations, and much higher aggregate throughput with respect to the case with a single active cell. This is true both with and without D2D, although a cooperative D2D approach clearly outperforms the case with no D2D. 
%Without ABS (as in Fig.~\ref{fig:heatmap_all} and Fig.~\ref{fig:heatmap_all5}), the system can still achieve outstanding results in terms of aggregate throughput, with and without D2D relay, although with $U_c=1$ the case with only $3$ active base stations yields more aggregate throughput than without ABS. {\it Most importantly, the distribution of throughputs over cells is very much unfair in any ABS state}, which justifies the effort to introduce a mechanism to enforce fairness.
%
%
The example of Fig.~\ref{fig:heatmaps} also shows that opportunistic mmD2D relay not only boost throughput, but also attenuates the difference achieved under various ABS states, which means that, with mmD2D relay, ABS can be more likely used to pursue other goals rather than simply transmission efficiency. For instance, since ABS can be seen as a mechanism that schedules base station activity, it is natural to think of ABS as of a tool for enforcing fairness by alternating system states conveniently. 

The above considerations motivate the problem formulation that will be formally presented in 
Section~\ref{s:optimization} 
%Section~\ref{s:cell_system_analysis}
in terms of optimal probabilities to select ABS states given user distribution and given the fact that users help each other with mmD2D relay. 
%Before proceeding with the formulation, in the next subsection we define the basic analytical tools to evaluate cellular performances in terms of resource utilization and system throughput under the effect of ABS, which are key to formulate the optimization problem.

\section{Proportional Fairness Optimization}
\label{s:optimization}

% Thanks to such analysis, we compute the ABS pattern which maximizes the (long term) system throughput under a vast set of heterogeneous conditions. Finally, we design an easy-to-deploy stochastic ABS mechanism that jointly achieves (short terms) high transmission efficiency as well as maximal user fairness according to the classical concept of {\it proportional fairness}.

%As we have shown in Section~\ref{s:analysis}, in a real network, both the asymptotic throughput and the instantaneous throughput achieved by the groups change due to two factors: $(i)$ the position of the users, and $(ii)$ the ABS state wherein the system lays. 
%
% With the above we have analytically characterized the {\it average} performance of a D2D-enabled  system in any possible ABS state. The analysis shows that ABS strategy should be adapted to changes in user spatial distributions, and that the main objective of ABS should be not only increasing aggregate system throughput, but also improve a throughput fairness metric, e.g., proportional fairness. Moreover, in practical and realistic scenarios, not only asymptotic spatial distributions of groups are not known, but also channel conditions can change and be not stationary, thus affecting the optimal strategy to select ABS patterns.
%

Here we first characterize the stochastic ABS pattern that achieves asymptotic maximal user fairness according to the classical concept of {\it proportional fairness}, which holds under a vast set of heterogeneous conditions. Under the assumption of ultra-dense scenarios, such an asymptotic analysis also approximates well the normal behavior of the network, since density conditions do not change over time. 

Afterwards, for highly dynamic scenarios wherein user density can fluctuate over time and the ultra-dense assumption cannot be used, we design an easy-to-deploy stochastic ABS mechanism that jointly achieves high transmission efficiency as well as maximal user fairness, closely following the variations of channel qualities and group locations in the system.

Before proceeding, note that, both with asymptotic and dynamic optimization, implementing the resulting stochastic ABS patterns has a twofold advantage: $(i)$ random patterns do not incur systematic discretization issues that might arise with deterministic allocations of states, and $(ii)$ they make it possible to generate distinct patterns for distinct base stations independently, thus reducing the complexity of {\it network controllers} issuing the patterns. 
Note that the formation of groups could be optimized as well, as it emerges from the analysis. However, it is left out of this work since it would require a more comprehensive study on user behavior, trustworthiness, costs and incentives, and other aspects that deserve a stand-alone project.

\subsection{Asymptotic Proportional Fairness Optimization}
 %Given the ABS state $s$, we have shown in Section~\ref{ss:asymptotic} how to compute the asymptotic throughput achieved by each group in the system, even in heterogeneous conditions. $E[\Gamma_c^s]$ is indeed the number of bits per second transmitted to group $c$, taking into account the particular spatial distribution of the group and all the boundary conditions. 
 The ABS mechanism allows switching among different states, and a group $c$ will obtain  average throughput $E[\Gamma_c^s]$ with probability $P^s$. 
 %so to achieve as performance the average of what achieved over the selected states. 
%Specifically, denoting by $P^s$ the fraction of time during which state $s$ is enforced, the resulting asymptotic throughput $E[\Theta_c]$ of group $c$ is simply given by 
% %
% \begin{equation}
% \vspace{-2mm}
% \label{eq:avg_throughput_asym}
% E[\Theta_c] = \sum_{s  \in \mathcal{S}} P^s E[\Gamma^s_c]. 
% \end{equation}
%
% $P^s$ can be re-interpreted as the probability mass distribution of ABS states, and 
% {\it the analysis unveils that performance is strongly affected by the set of ABS states used and their probabilities rather than by the order in which such states are visited}.
% %what affects performance is the set of ABS states used, not the order in which such states are visited. 
Based on such insight, we stochastically approach the ABS selection problem, i.e., we assign probabilities to select ABS states, and use such probabilities to generate sequences ({\it patterns}) of ABS states randomly. Hence, proportional fairness in terms of throughput is achieved by optimizing ABS state probabilities.
%
% The optimization problem that we formulate in order to select such ABS state probabilities goes as follows:
%
\begin{align}
\label{pr:be-distr_asympt}
& \textbf{Problem } \texttt{Asymptotic ABS-PF}: \nonumber\\
& \text{Select $P^s, \forall s\in\mathcal{S}$, so to:} \nonumber\\ 
& \begin{array}{rl}
\text{maximize }  & \widehat{\eta}_a =\sum\limits_{c \in \mathcal{C}} w_c \; \log \left( \sum\limits_{s \in \mathcal{S}} P^s E[\Gamma_c^s] \right); \\
\text{subject to: } & \sum\limits_{s\in\mathcal{S}} P^s = 1,\vspace{1mm}\vspace{-1mm}\\
		         & P^s \in [0,1], \quad \forall s \in \mathcal{S};
\end{array}
\end{align}%
where 
%State probabilities $P_s$ are the decision variables of the problem and affect the utility function through~\eqref{eq:Theta}, and their sum has to be $1$. As for 
weights $w_c$ are used to tune the group fairness.
%normalized to the number of users attached to the same base station. This guarantees fairness between users of different cluster sizes, as a weighted policy is applied. 
Since the argument of the $\log$ function in the maximization is linear with the decision variables $P^s$, the problem is convex and admits a global optimum that can be found with any off-the-shelf solver. Moreover, by linearizing the problem (e.g., the $\log$ can be approximated by a polygonal chain), the optimum can be found in polynomial time. 

Once the probabilities $P^s$ are computed, the node running the optimization (i.e., a {\it network controller}) stochastically builds and distribute a stochastically optimal ABS pattern by choosing, for each subframe composing the ABS pattern, a state at random according to optimal probabilities. The assigned ABS pattern is then repeated indefinitely at the gNB.

\subsection{Dynamic Proportional Fairness Optimization}
\label{s:dynamic_opt}
We next focus on a highly dynamic evolving system and we provide a mechanism to optimize ABS patterns to achieve proportional fairness over time, accounting for network dynamics as they are observed.  

To formulate our proportional fairness optimization problem, let us consider that, in short intervals of time of duration $T$, in which mobility effects are negligible, the throughput achievable in each state $s$ by each group does not change and can be indicated as the group throughput computed at any point in time within that interval. Thus, considering a time-slotted optimization framework starting at time $t_0$, composed by intervals $I_n~\triangleq~[t_0+nT; t_0+(n+1)T), \; n \ge 0$,  we can denote the throughput as $E\left[\Gamma_c^s(I_n)\right] = E\left[\Gamma_c^s (t)\right]$ computed with the group positions evaluated at any $t$ chosen in $I_n$.
% as in~\eqref{eq:throughput}, though the meaning of $t$ in this case is slightly different. 
Also in this dynamic version of the ABS optimization, during intervals $I_n$, it is however possible to chose subsequently different ABS states, so to achieve as performance the average of what achieved over the selected states. We denote with $P^s(I_n)$ the fraction of time during which state $s$ is enforced in interval $I_n$, the resulting throughput 
%$\Theta_c$ 
of group $c$ in interval $I_n$ is: 
\begin{equation}
\vspace{-2mm}
\label{eq:avg_throughput}
%\Theta_c (I_n) 
E \left[\Gamma_c (I_n) \right] 
= \sum_{s  \in \mathcal{S}} P^s(I_n) E \left [ \Gamma^s_c(I_n) \right ] . 
\end{equation}

% and choosing the correct sequence of ABS states.
%propose a {\it stochastic} approach to properly choose ABS patterns formulation in which ABS states are chosen according to a probability distribution that optimizes proportional fairness for the throughput achieved by the users. 
%A sequence of ABS states defines the ABS pattern issued to each base station within the intervals $I_n, n \ge 0$, where Eq.~\eqref{eq:avg_throughput} still holds with $P^s(I_n)$ re-interpreted as the probability mass distribution of ABS states.   
%The actual ABS states selected (and their order) form the ABS patterns that the base stations have to execute within intervals $I_n, n \ge 0$, and Eq.~\eqref{eq:avg_throughput} holds with $P^s(I_n)$ re-interpreted as the probability mass distribution of ABS states.   
%
%In practical and realistic scenarios, not only asymptotic spatial distributions of groups are not known, but also channel conditions can change, affecting the ABS pattern selection procedure. 
Clearly, the order in which ABS states are visited is not important and the computation of such ABS state probabilities must be repeated every interval $I_n$, due to network dynamics. The choice for the duration $T$ of such interval is pivotal for system performance: a {\it short} interval allows to consider the network as static, while a {\it long} interval accounts for including several ABS states, which in turn increases the accuracy resulting from implementing optimal probabilities with a finite-length ABS pattern.
%, so that the frequency of state occurrences in $T$ will actually converge to the chosen probability.
%whose duration $T$ has to be {\it short enough} to allow to consider the network as static and {\it long enough} to include several ABS states, so that the frequency of state occurrences in $T$ will actually converge to the chosen probability. 
For instance, the probabilities could be chosen once per second, in line with normal ABS decision-making procedures, which involves patterns of tens or hundreds of states wherein each state lasts at least 1 ms. 
%
%Note that implementing stochastic ABS patterns has a twofold advantage: $(i)$ random patterns do not incur systematic discretization issues that might arise with deterministic allocations of states, and $(ii)$ they make it possible to generate distinct patterns for distinct base stations independently, thus reducing the complexity of {\it network controllers} issuing the patterns.

The optimization problem that we formulate in order to select such ABS state probabilities is based on a long-term proportional fairness metric, in which the throughput is observed over a period of $p$ past intervals and predicted for the next interval $I_n$ (of course, the decision made at each point in time only affects $I_n$). 
To achieve so, the optimization is repeated at the begin of each interval $I_n$, and we define a utility function $\widehat\eta$ based on the $\log$ of group throughputs (to introduce proportional fairness) computed over $p+1$ intervals:
\vspace{-2mm}
\begin{align}
\label{pr:be-distr}
& \textbf{Problem } \texttt{Dynamic ABS-PF}: \nonumber\\
& \text{At time $t=t_0+nT$, select $P^s(I_n), \forall s\in\mathcal{S}$, so to:} \nonumber\\ 
& \begin{array}{rl}
\text{maximize }  & \widehat{\eta} =\sum\limits_{c \in \mathcal{C}} w_c \; \log \left( \sum\limits_{k=n-p}^n \alpha_{n-k} E\left[\Gamma_c(I_k) \right] \right); \\
\text{subject to: } & \sum\limits_{s\in\mathcal{S}} P^s(I_n) = 1,\vspace{1mm}\vspace{-1mm}\\
		         & P^s(I_n) \in [0,1], \quad \forall s \in \mathcal{S};
\end{array}
\end{align}%
where 
%State probabilities $P_s$ are the decision variables of the problem and affect the utility function through~\eqref{eq:Theta}, and their sum has to be $1$. As for 
weights $w_c$ are used to tune the group fairness target and coefficients $\alpha_k$ define how past samples of throughput affect future decisions.
%normalized to the number of users attached to the same base station. This guarantees fairness between users of different cluster sizes, as a weighted policy is applied. 
Since $E\left[\Gamma_c (I_n)\right]$---which is the only unknown term in the sum inside the $\log$ argument (because past values have been observed)---is linear in the decision variables $P^s$, also the above-defined dynamic version of the optimization problem is convex, admits a global maximum and can be easily linearized and solved in polynomial time.

Every time probabilities $P^s(I_n)$ are computed, the network controller stochastically builds and distributes a new ABS pattern by choosing, for each subframe composing the ABS pattern, a state at random according to new optimal probabilities. Such ABS pattern is valid until a new pattern is issued. 

%Note that the optimization problem ABS-PF jointly manages ABS and intra-group transmissions. Indeed, the transmission efficiency achieved in each ABS state affects the ABS pattern deployed by the base station (via the optimization ABS-PF and the analysis proposed in Sec.\ref{transmission_efficiency}), while, in turn, the ABS pattern affects the average transmission efficiency achieved by the groups. 

\subsection{Remarks on the Stochastic Optimization of ABS Patterns}

In both Problems~\texttt{Asymptotic ABS-PF} and \texttt{Dynamic ABS-PF}, weights $w_c$ can be selected based on the desired fairness target. E.g., 
%if all weights are the same, the optimization will try to equalize group throughputs as much as possible. Instead, 
for targeting equal throughput on a {\it per-user} basis, given that the group throughput is equally shared by group members, $w_c$ can be set as the number of users forming group $c$, so that group throughputs will be as much as possible proportional to group sizes. 

Coefficients $\alpha_k$ in Problem~\texttt{Dynamic ABS-PF} can be taken as a non-decreasing sequence of non-negative weights, so that past values of the throughput receive less or equal importance with respect to the prevision for next interval $I_n$. For example, exponentially decaying coefficients or constant coefficients represent simple and widely adopted solutions for this kind of digital filtering problems.  

%Furthermore, since the values of $\Theta_c (I_k), k<n,$ are known at optimization time for all groups, 
%the argument of the $\log$ function in $\widehat\eta$ consists of a constant term representing the past, and the throughput foreseen for interval $I_n$, which is linear with respect to $P^s(I_n)$ values. Since the log can be linearized through discretization, and its argument is linear, Problem~\texttt{ABS-PF} can be thereby straightforwardly linearized as well, so any off-the-shelf solver is suitable to find the optimal values for $P^s(I_n)$. 
%Once the probabilities $P^s(I_n)$ are computed, the node running the optimization (i.e., a network controller), can stochastically build ABS patterns to be distributed to the base stations by choosing at random, for each subframe composing the ABS pattern, an ABS state according to the optimal probabilities.

Interestingly, Problem~\texttt{Dynamic ABS-PF} is simple to adapt also to cover the case in which the traffic of groups is not saturated. In such a case, the time window $(p+1)T$ has to be smaller than the interval during which the set of receivers changes. In fact, during such interval all active receivers can be considered as saturated and the presented analysis holds.

\section{Practical Details}
\label{s:implementation}

Before proceeding with the numerical assessment of our proposal, here we comment on a few practical details that have to be tackled in order to implement asymptotic and dynamic optimization of ABS patterns.

\subsection{Groups}
\label{ss:groups}
The formation and presence of groups of users leaning toward cooperation is key for the success of opportunistic relay approaches. Groups might form using services like Google Nearby, which is an Android feature to discover D2D peers and request connection~\cite{CCB18}. Other groups might form by static user configuration, e.g., by pre-authorizing communication between devices belonging to the same owner (like it happens for wireless mice and hands-free speakers). 

In any case, the group has an effective role in the system only if  relay opportunities last a sufficiently large period of time, so that the group setup overhead can be neglected. 

% Therefore, we assume to exploit the groups which are intrinsically present in our access network and, as a consequence, how to form groups is out of the scope of this contribution. In order to exemplify our assumption, in the following we use the Reference Point Group Mobility (RPGM) model ~\cite{Mobility_Gerla} for our analysis. We define a center of gravity for the group, which imposes group motion parameters, such as speed or direction. 

\subsection{Resource Allocation with Groups}
%To pursue transmission efficiency and system throughput, users are allowed to use a group member as relay node to directly interact with the reference base station. 
Groups are served by the gNB like they were normal UEs. For the sake of fairness, the gNB should use a weighted round robin policy and assume the number of users within the group as the scheduling weight. Therefore, the average number of cellular resources allotted to each user would remain constant considering all possible group configurations. 

\subsection {Relay Node}
%With the ultra-high capacity of eD2D, the system throughput bottleneck is represented by the RAN domain. %cellular transmission rate only. 
%Therefore, the more efficient cellular transmissions are, the higher the system throughput grows. 
%Thereby it is key not only to 
Selecting relay nodes opportunistically and switching them swiftly, as soon as channel conditions vary, is a key-enabler for mmD2D.  We assume that the relay node is selected by the gNB on a per-packet basis, leveraging CSI reports, so that cellular transmissions always occur on the strongest cellular channel. A similar approach has been suggested and experimentally validated in~\cite{AMG16} for WiFi Direct.
%over the 2.4-GHz ISM band.
Please note that continuous re-election of relay nodes has no practical drawbacks on end-to-end latency due to the huge available bandwidth. %from the point of view of end-to-end traffic flows management. In fact, 
%as relay decisions are centralized taken %on a per-packet basis 
%by the base station (who knows which group member has the better channel), after which, relay nodes serve the aggregate group data demand via the MAC of WiGig, thus experiencing no practical delay thanks to its huge available bandwidth. 
%\vs{DO WE NEED THE FOLLOWING? I WOULD SHORTEN THIS SUBSECTION} The described relay scheme operates at PHY and MAC level, so that, in our system, data forwarding is transparent to protocols such as TCP or UDP.

\subsection{ABS {Pattern Generation}}

%ABS is a 3GPP time-domain scheduling scheme to prevent macro base stations from transmitting data in a particular set of subframes. Blanking decisions are taken based on interference suffered in the system, through a binary pattern that specifies whether a subframe has to be blanked or not ({\it ABSF pattern}). 
3GPP standard guidelines allow to implement the ABS scheme in a conventional cellular system without imposing any constraint on the specific set of subframes to blank. Standard specifications describe an ABS {\it application ratio} defined as the number of used subframes over the total number of subframes within a pattern. Once this fixed ratio is imposed by a {\it network controller}, base stations may make random choices to select the specific pattern of subframes to be blanked.

\section{Performance Evaluation}
\label{s:eval}

In this section, we evaluate our proposals---hereafter indicated as \texttt{Asymptotic ABS-PF} and \texttt{Dynamic ABS-PF} after the names of the optimization problems defined in Section~\ref{s:dynamic_opt}---for several cellulars scenarios, including different mobility behaviors, group sizes and network area densities. In particular, we show that:
\begin{itemize}
\item the analytical model derived in Section~\ref{s:analysis} provides accurate results as validated against a packet-level simulator;
\item the analytical model yields reliable results even when some of our assumptions do not hold (for instance, user groups are not collapsed in one spatial point) and under truly simulated user mobility models;
\item our practical optimization solutions provide outstanding results in terms of throughput, transmission efficiency and user fairness, when compared to state-of-the-art approaches;
\item \texttt{Asymptotic ABS-PF} and \texttt{Dynamic ABS-PF} outperform current standardized solutions when applied to optimize network operations in realistically evaluated scenarios, such as a dense-urban scenario covered with a heterogeneous cell deployment. 
%in a urban area.
\end{itemize}
%\vspace{-0.5mm}
%
%present a numerical evaluation for the model derived in Section~\ref{s:analysis} and we validate it against a packet-level simulator written in MATLAB. With this, we can properly validate our model when some assumption does not hold (for instance, user groups are not collapsed in one spatial point) and the user mobility is truly simulated.  
%Additionally, we show the effectiveness of our practical solution by carrying out a comparison with state-of-the-art approaches. 
%

\subsection{Experimental Setup}

We have developed a Matlab event-driven packet simulator to study regular deployments, namely synthetic scenarios. We consider two radically different scenarios to show that our approach brings substantial benefits under very different operational conditions: $(i)$ a simple cellular network with $7$ base stations using the same transmission power ($250$mW) and regularly spaced with Inter-Site Distance (ISD)  equal to $50$m in a $150$m~$\times$~$150$m area, and $(ii)$ the heterogeneous O2 deployment used for the example discussed in Section~\ref{s:example}.
Only downlink transmissions are taken into account in our simulation, with a $20$MHz bandwidth and with the pathloss model described in~\cite{TR_942}. 
%Nevertheless, uplink transmission issues may be similarly tackled by applying our approach. 
Transmission queues are fully backlogged. Relay group $c$ can include up to $U_c$ users: whenever $U_c\!=\!1$ we will be dealing with the case of no cooperative mmD2D communications. 
When active, base stations apply a weighted round robin policy to deliver the offered traffic to the relay groups, using the sizes of the groups as weights. We use the Random WayPoint (RWP) mobility model~\cite{waypoint} to move the center of gravity of each group within the simulated area, with a speed ranging from $1$ to $10$ m/s. 
For solving Problem~\texttt{ABS-PF} presented in Section~\ref{s:dynamic_opt}, we run the ABS optimizations every $T=500$ ms, with weights $w_c = U_c$ to achieve per-user fairness, $\alpha_k=1$, and $p=20$. 
%Patterns are built and distributed at that pace (such patterns are repeated cyclically by the base stations for $500$ ms). 
Network simulations last $500$ s, whereas user channel conditions are evaluated on a subframe basis, e.g., each $1$ ms. 
%
%Additionally, we evaluate our approach when applied to a realistic and heterogeneous base stations deployment in London city, UK. 
For the case of the O2 deployment in London, we deal with a high-dense area in which the RWP mobility is developed through the streets on the map. 
%
%We show the gain of applying our practical approach rather than keeping current legacy solutions. 
All presented results are provided 
%as the average resulting from 
with $95\%$ confidence intervals.
%$25$ simulation instances. 
%have been carried out to compute each of the reported values.

To assess the performance of our practical solutions, we evaluate \texttt{Asymptotic ABS-PF} and \texttt{Dynamic ABS-PF} in terms of system throughput and fairness, the latter being measured by means of the well-known Jain's Fairness Index (JFI). For the sake of comparison, we also consider a solution without ABS, namely \texttt{Legacy}, as well as another stochastic approach, using an optimization problem similar to \texttt{Dynamic ABS-PF} that maximizes system throughput rather than user fairness, i.e., it optimizes the sum of group throughputs rather than the sum of logarithms. We refer to such policy as \texttt{Max Throughput}. We further benchmark our approach against two practical state of the art heuristics: \texttt{BSB},  proposed in~\cite{TWC_sciancalepore_2016}, in which ABS is used to target a max-min utility function; and \texttt{DRONEE} \cite{asadi2014dronee}, in which relay groups (mmD2D groups in our case) are formed dynamically to improve throughput.
Lastly, we compare \texttt{Dynamic ABS-PF} to standard randomized ABS schemes with different ABS application ratios.

\subsection{Model Validation}
\label{model_validation}
In order to evaluate and validate the analytical model presented in Section~\ref{s:analysis}, in Fig.~\ref{fig:validation} we graphically provide a set of analytical results in terms of system throughputs $\Gamma^s$. For the validation, we use static scenarios wherein the ABS state and the position of nodes remain unchanged, so to compare the analytically derived throughputs (one for each state $s$) with long-run averages per each state observed in simulations. 
%(corresponding to $\hat{\Gamma}|_s$ when considered in saturation, as stated in Section~\ref{s:cell_system_analysis})
%by varying the applied ABS state $s$. 
Due to the huge computational effort required for simulating every possible ABS states, only some significant ABS states have been considered within the packet simulator. 
We mark with a red circle the ABS state corresponding to all base stations simultaneously active, to point out the impact on the system throughput of no ABS application.
%how no ABS application impacts .
Notably, we observe that simulation results closely follow those provided through the analysis, properly validating the accuracy of our study with and without mmD2D relay groups.

\ifsingle
\begin{figure} [!t]
\centering
	\includegraphics[width=0.50\textwidth]{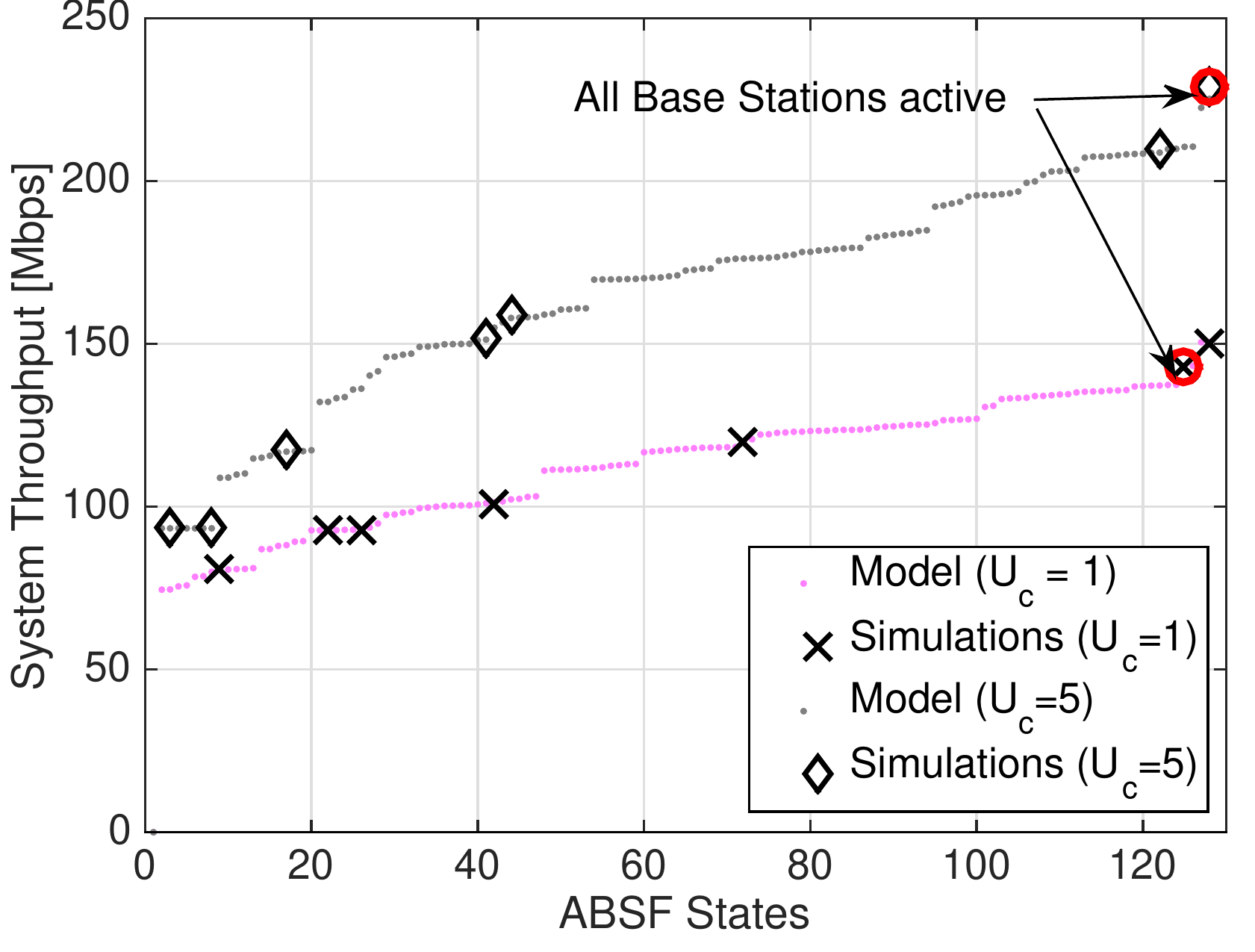}
	\caption{Model validation through exhaustive simulations.}
	\label{fig:validation}
\end{figure}
\else
\begin{figure} [!t]
\centering
	\includegraphics[width=0.3\textwidth]{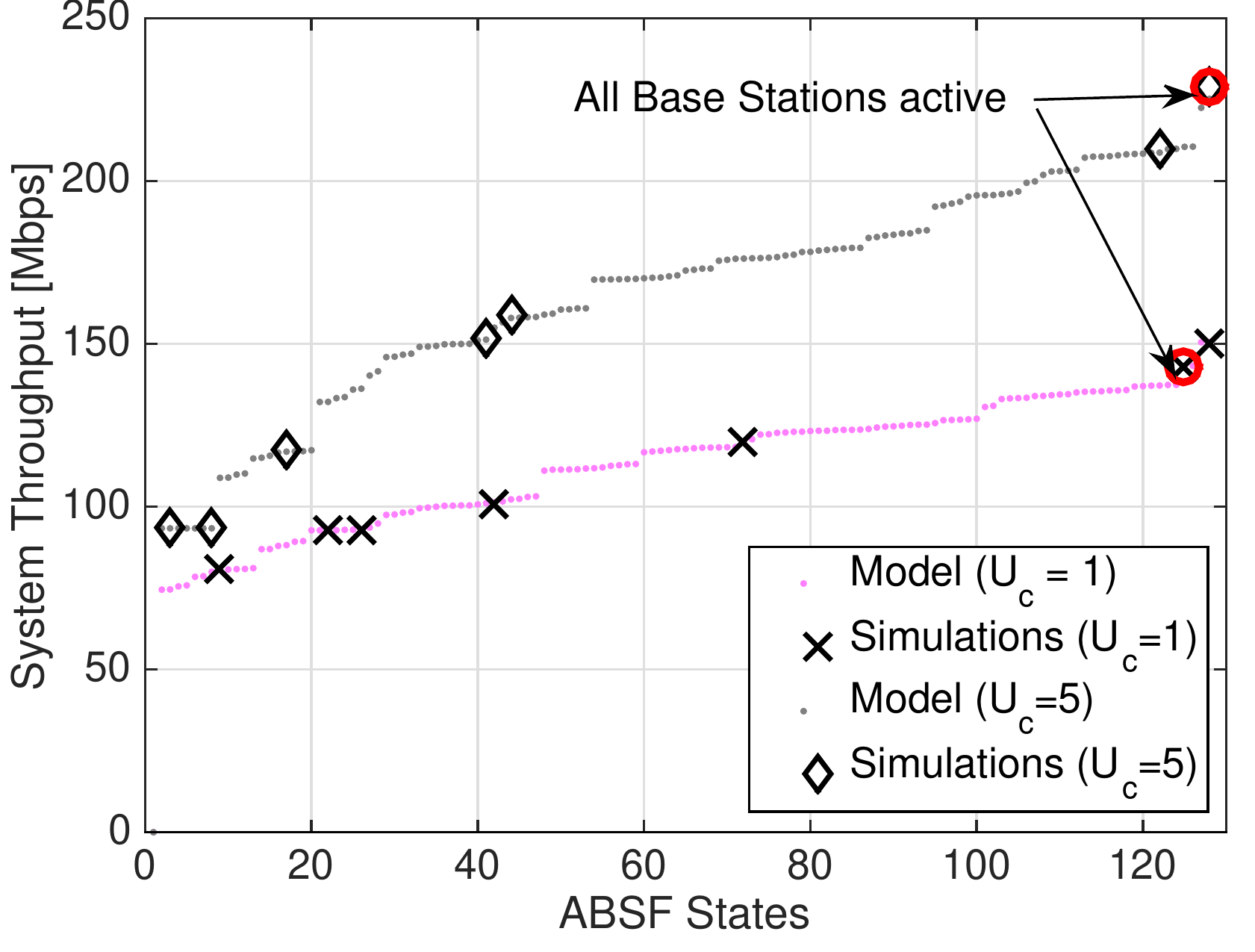}
	%\vspace{-3mm}
	\caption{Model validation through exhaustive simulations.%\vm{possiamo aggiungere qualche punto di simulazione?}
	}
	\label{fig:validation}
%	\vspace{-5mm}
\end{figure}
\fi

%\subsection{Optimal ABS Pattern Selection}
\subsection{Performance in Homogeneous Deployments}
\label{ss:optimal_selection}

We next assess the performance of our practical solution for synthetic deployments of regularly spaced base stations using the same transmission power. 
%through a stochastic ABS scheme adopting the solution of Problem~\texttt{ABS-PF}. We show the system performance in terms of aggregate system throughput and in terms of fairness, expressed by the well-known Jain Fairness Index
%%,\footnote{The Jain Fairness Index is used to determine whether the system resources are fairly shared. It can be calculated as JFI = $ \left (\sum\limits_{i=1}^n x_i \right )^2 \Big / \left ( n \sum\limits_{i=1}^n x_i^2 \right )$.} 
%(JFI). For the sake of completeness, for our evaluation we also consider a system without ABS, namely \texttt{Legacy}, as well as another stochastic approach, similar to Problem~\texttt{ABS-PF}, which optimally maximizes the total system throughput rather than the user fairness, i.e., the $\log$ function has been removed from the objective function in \eqref{pr:be-distr}. We refer to such policy as \texttt{Max Throughput}. Lastly, we benchmark our approach against a practical state of the art heuristic, namely \texttt{BSB},  proposed in~\cite{TWC_sciancalepore_2016}, in which a max-min utility function is targeted.

Figs.~\ref{fig:throughput} and \ref{fig:fairness} present system throughputs and fairness levels achieved with different ABS policies. 
Results are drawn for a few examples of user populations and numbers of relay groups. In the figures, the x-axis reports the number of considered users, and, when applicable, the number of relay groups. Each group $c$ consists of $U_c$ users, where $U_c$ is a uniform random variable drawn between $1$ and $5$. Note that, figures report two cases with $150$ users, and two cases with $300$ users, i.e., with and without groups. Therefore, it is easy to observe both the impact of the user population size as well as of D2D communications. 
%the way the same users could collaborate through D2D communications. 

In Fig.~\ref{fig:throughput}, the \texttt{Legacy} scheme shows an acceptable level of throughput even when compared to the \texttt{Max Throughput} scheme, but  only when mmD2D is not used. Our \texttt{Asymptotic ABS-PF} and \texttt{Dynamic ABS-PF} schemes provide reasonable results in absence of mmD2D, outperforming the heuristic provided by~\texttt{BSB} and achieving similar throughputs as the \texttt{Legacy} scheme. Notably,  in this homogeneous scenario, \texttt{Asymptotic ABS-PF} obtains even better throughput than \texttt{Dynamic ABS-PF}.
However, both stochastic ABS patterns do not help much in terms of throughput, unless mmD2D is enabled. This confirms that ABS, even when optimized, is not a suitable solution on its own. Instead, in combination with mmD2D, stochastic ABS patterns make the difference.  
Moreover, \texttt{Legacy} and \texttt{BSB} schemes do not take advantage of mmD2D relay groups, and their performance figures only slightly change with the number of users and groups. 
In contrast, \texttt{Max Throughput} leverages the transmission efficiency enhancements due to opportunistic mmD2D relay and significantly boost throughputs. In all cases, \texttt{Max Throughput} represents the highest achievable network throughput. Therefore, Fig.~\ref{fig:throughput} reveals the net potential of mmD2D and the fact that little gain can be expected by any scheme unless mmD2D is jointly enforced.   

Note that, as visible in Fig.~\ref{fig:throughput}, user density plays a very minor role in terms of system throughput. Conversely, density has a huge impact on fairness, as evaluated in Fig.~\ref{fig:fairness}, wherein we do not report results for \texttt{Legacy} and \texttt{Asymptotic ABS-PF}, since those schemes are perfectly fair by definition in a completely homogeneous scenario like the one under analysis, at least on the long run.\footnote{Note that, in the homogeneous case, all users (groups) have the same spatial distribution. Thereby, \texttt{Legacy} as well as any ABS strategy that does not change/adapt over time---as it happens with \texttt{Asymptotic ABS-PF}---results in the same asymptotic throughput for all users (groups).} 
In the figure, \texttt{Dynamic ABS-PF} exhibits very powerful results when compared to \texttt{BSB} and \texttt{Max Throughput}. However, we need to remark that \texttt{BSB}, as the network becomes denser, shows better results in terms of fairness at the expenses of a very poor system throughputs. Nonetheless, \texttt{Dynamic ABS-PF} yields fairness levels very close to optimal fairness metrics in all cases (i.e., very close to $1$). This confirms that stochastically-issued ABS patterns and mmD2D in combination are suitable for achieving high fairness while improving throughput. 
    
\ifsingle
\begin{figure} [!t]
\centering
	\includegraphics[width=0.49\textwidth]{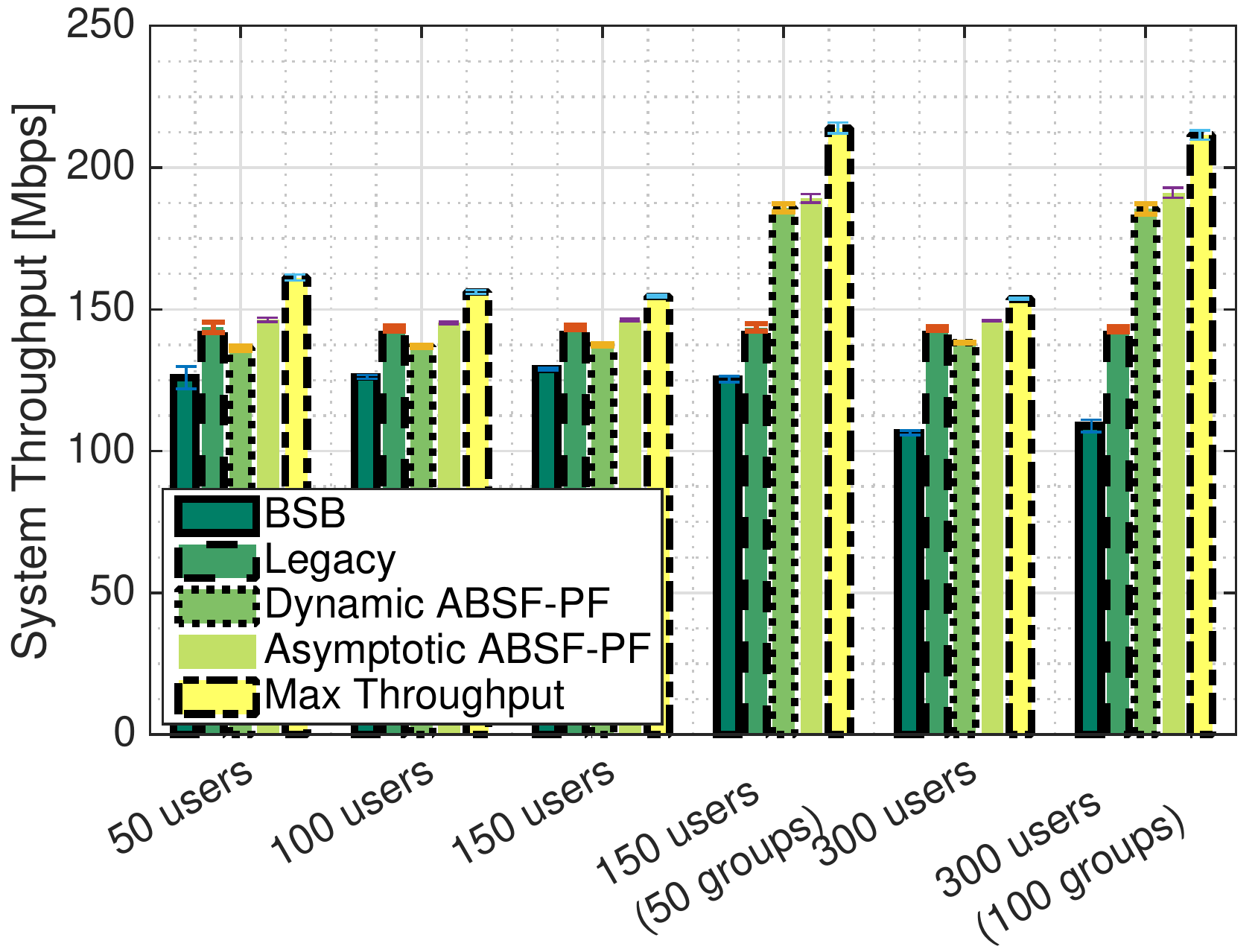}
%\vspace{-8mm}
	\caption{\footnotesize System throughput under different optimization policies.}
	\label{fig:throughput}
%\vspace{6mm}
	\includegraphics[width=0.49\textwidth]{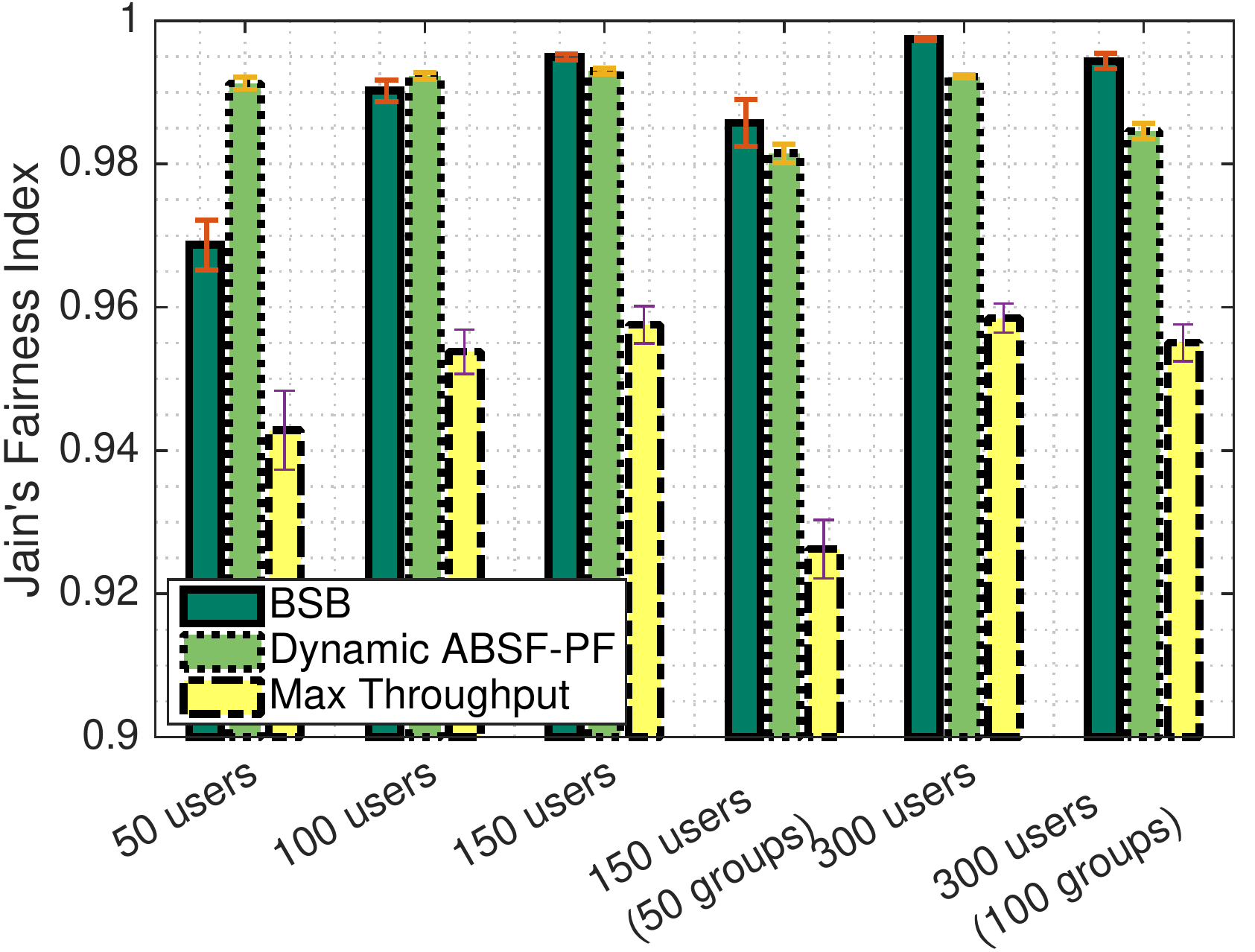}
%\vspace{-8mm}
	\caption{\footnotesize Jain's Fairness Index under different optimization policies.}
	\label{fig:fairness}
\end{figure}
\else
\begin{figure} [!t]
\centering
	\includegraphics[width=0.35\textwidth]{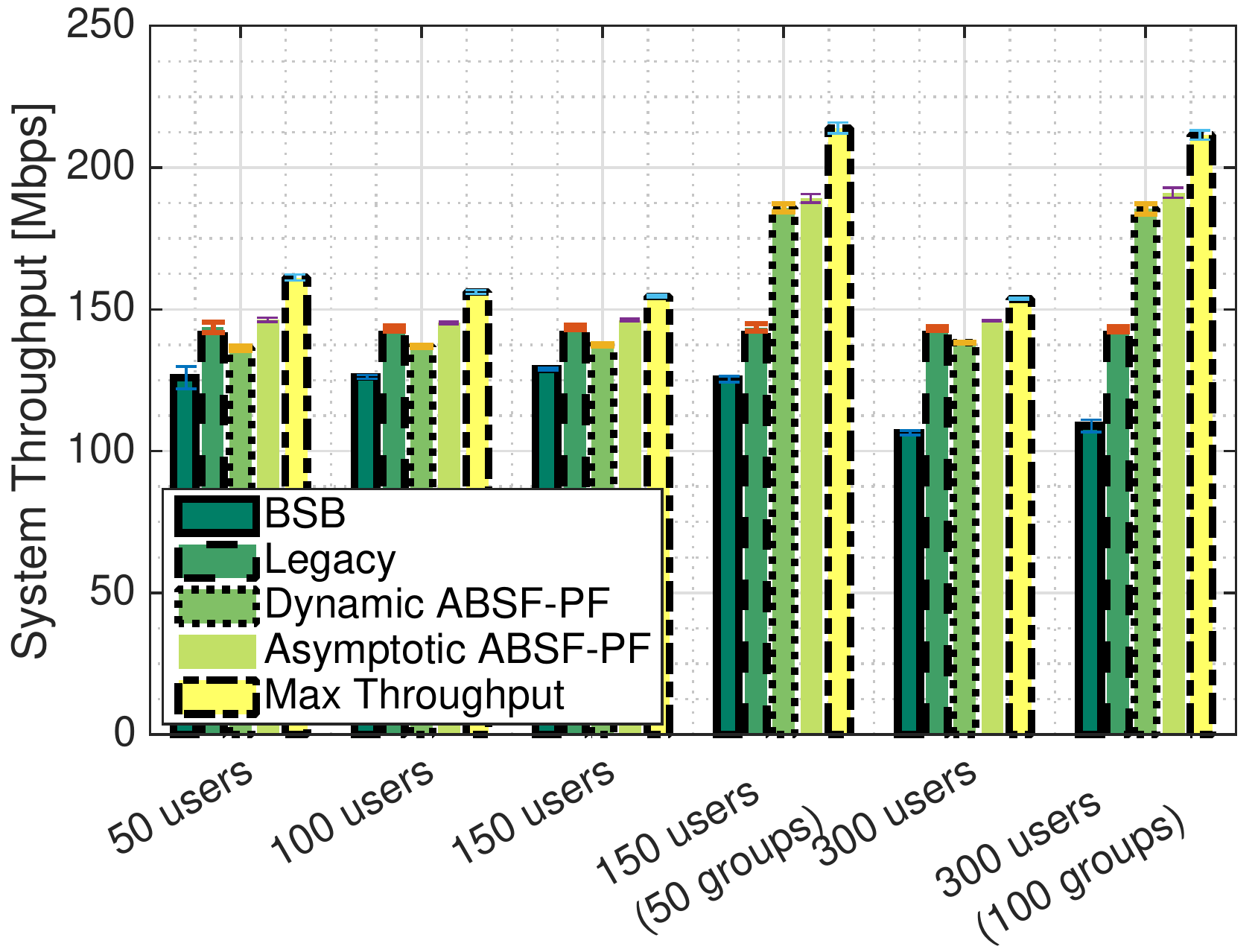}
%\vspace{-5mm}
	\caption{\footnotesize System throughput under different optimization policies.}
	\label{fig:throughput}
%\vspace{3mm}
	\includegraphics[width=0.35\textwidth]{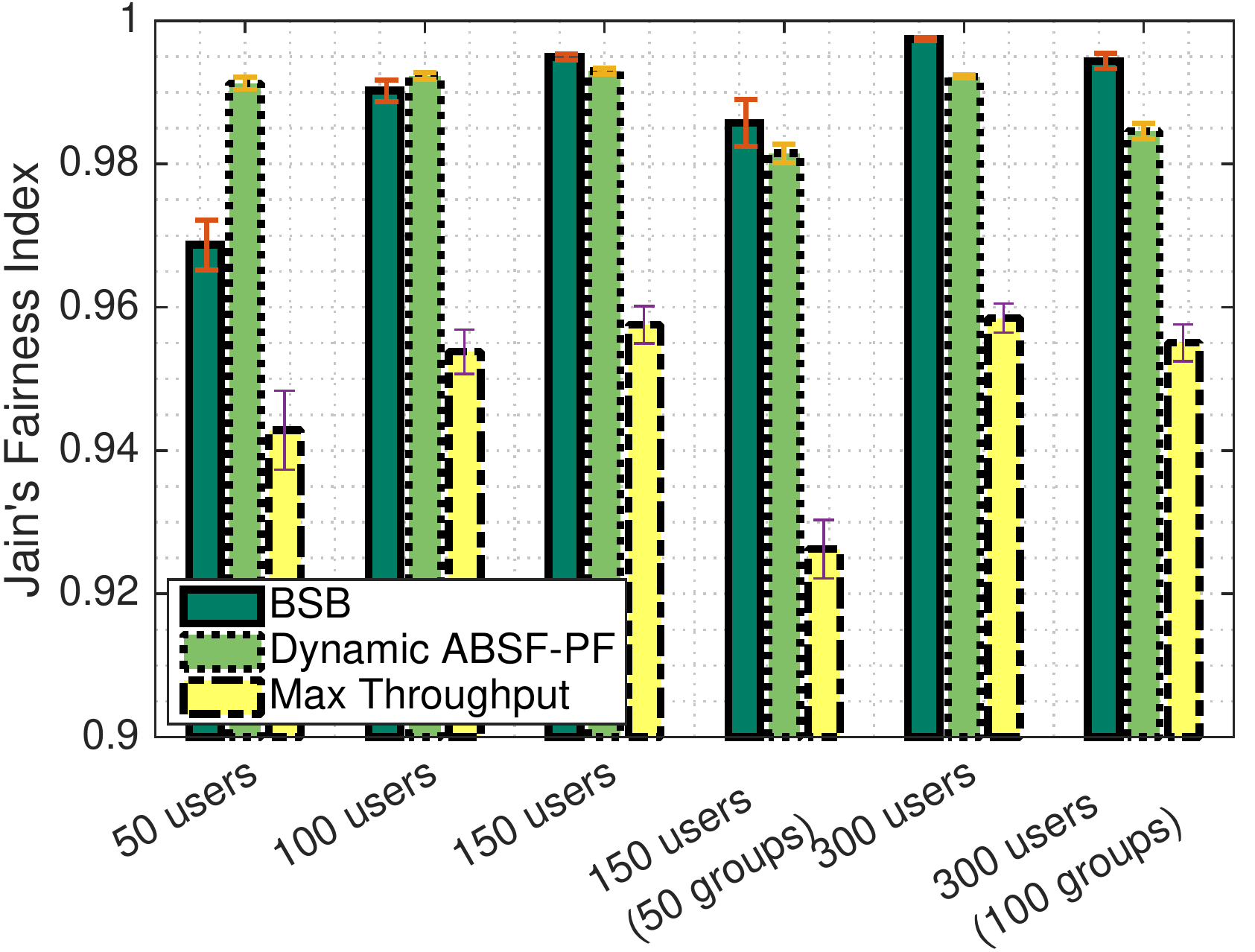}
%\vspace{-5mm}
	\caption{\footnotesize Jain's Fairness Index under different optimization policies.}
	\label{fig:fairness}
%\vspace{-6mm}
\end{figure}
\fi

Next we compare \texttt{Dynamic ABS-PF} to standard ABS implementation with typical fixed application ratios, ranging from $4/8$ to $6/8$ (i.e., blanking from $50\%$ to $25\%$ of the subframes at random), as suggested in~\cite{absf}. 
Here, blank subframes are randomly chosen by each base station independently, and the resulting patterns are automatically repeated to fill up the standard ABS pattern of $80$ subframes. The 80-subframe long pattern is then repeated indefinitely~\cite{absf}.     
Specifically, 
Fig.~\ref{fig:throughput_stocastic} compares system throughputs achieved with \texttt{Dynamic ABS-PF} and when three fixed ABS application ratios are applied to the system (with and without mmD2D relay). 
%The denser the network, the higher the spectral efficiency of our \texttt{ABS-PF} proposal. Conversely, 
The figure shows that standard ABS schemes do not bring significant throughput improvements. Moreover, when fairness issues are considered, in Fig.~\ref{fig:fairness_stocastic}, \texttt{Dynamic ABS-PF} exhibits strong advantages with respect to fixed ABS application ratios.

\ifsingle
\begin{figure} [!t]
\centering
	\includegraphics[scale=0.47]{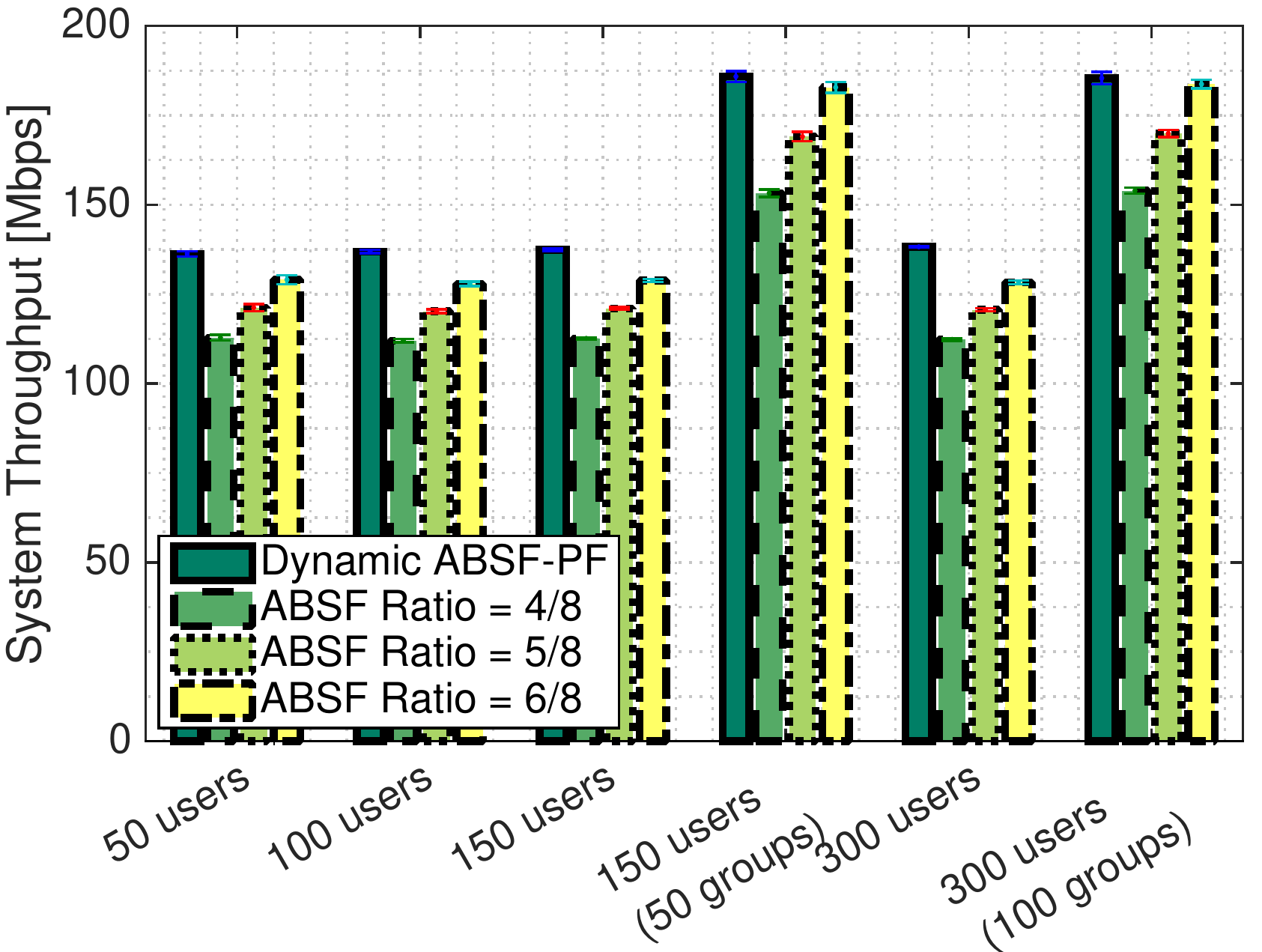}
%\vspace{-4mm}
	\caption{\footnotesize System throughput under different ABS schemes.}
	\label{fig:throughput_stocastic}
%\vspace{-3mm}
\end{figure}

\begin{figure} [!t]
\centering
	\includegraphics[scale=0.47]{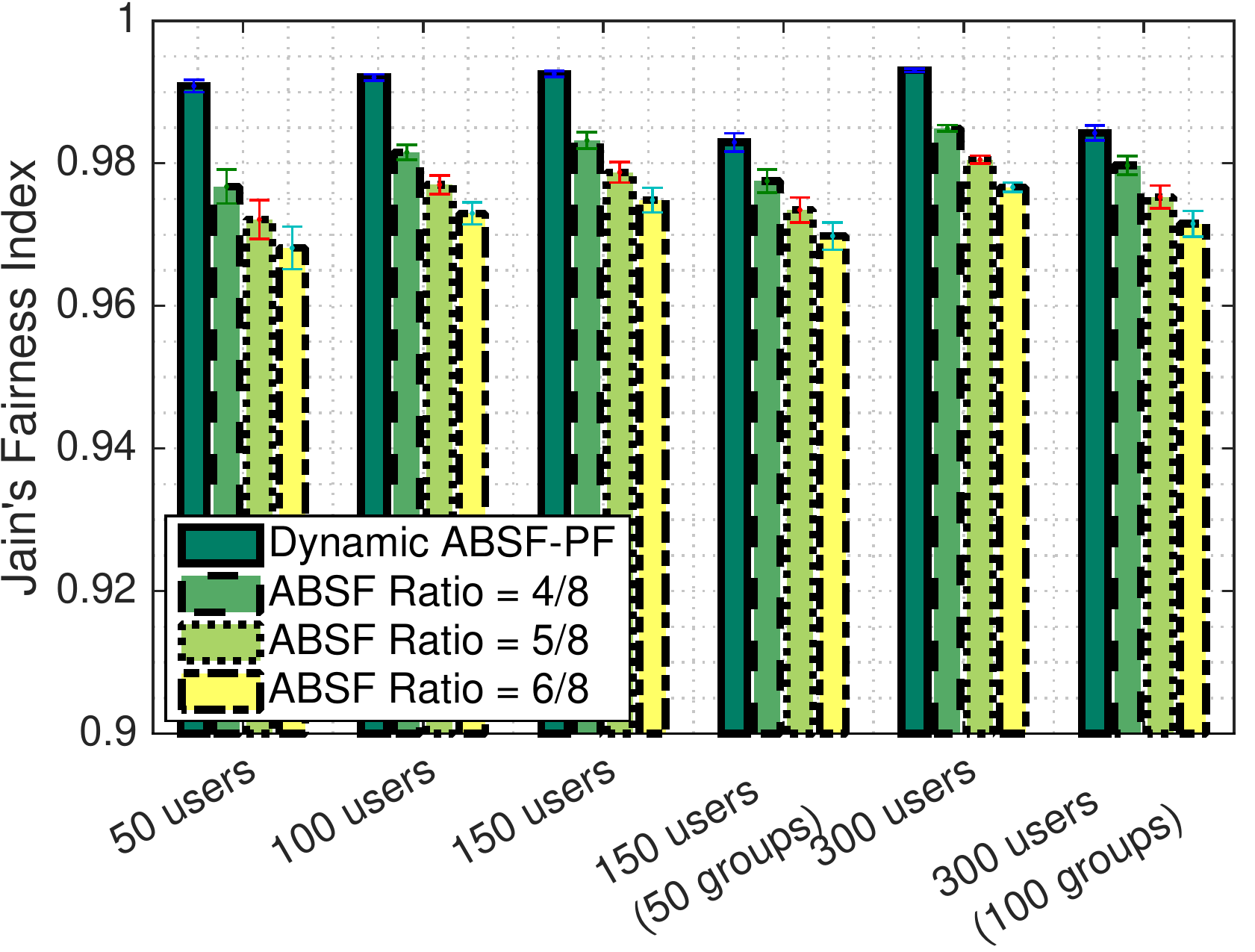}
%\vspace{-4mm}
	\caption{\footnotesize Jain's Fairness Index under different ABS schemes.}
	\label{fig:fairness_stocastic}
%\vspace{-5mm}
\end{figure}
\else
\begin{figure} [!t]
\centering
	\includegraphics[scale=0.31]{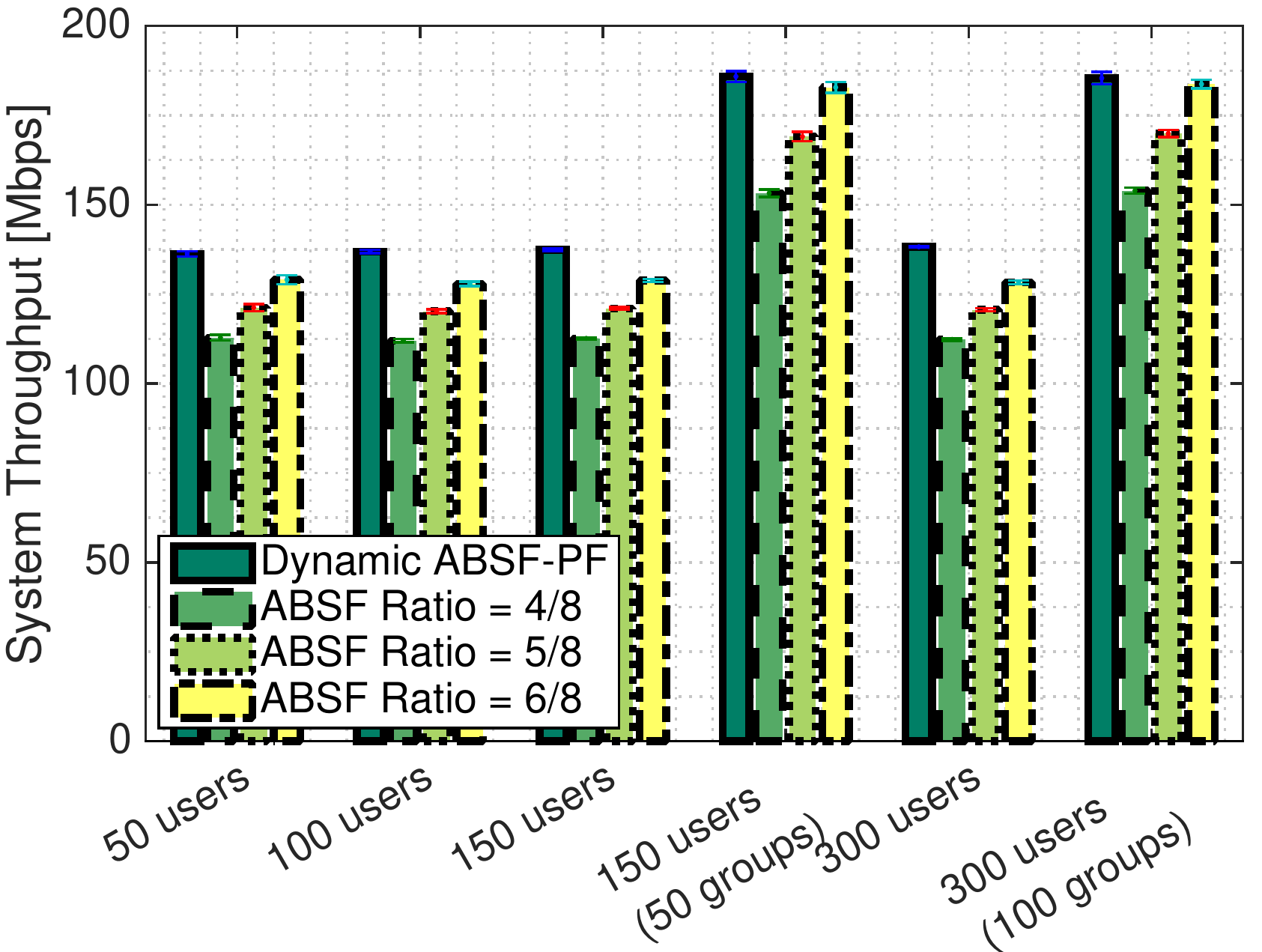}
%\vspace{-4mm}
	\caption{\footnotesize System throughput under different ABS schemes.}
	\label{fig:throughput_stocastic}
%\vspace{-3mm}
\end{figure}

\begin{figure} [!t]
\centering
	\includegraphics[scale=0.31]{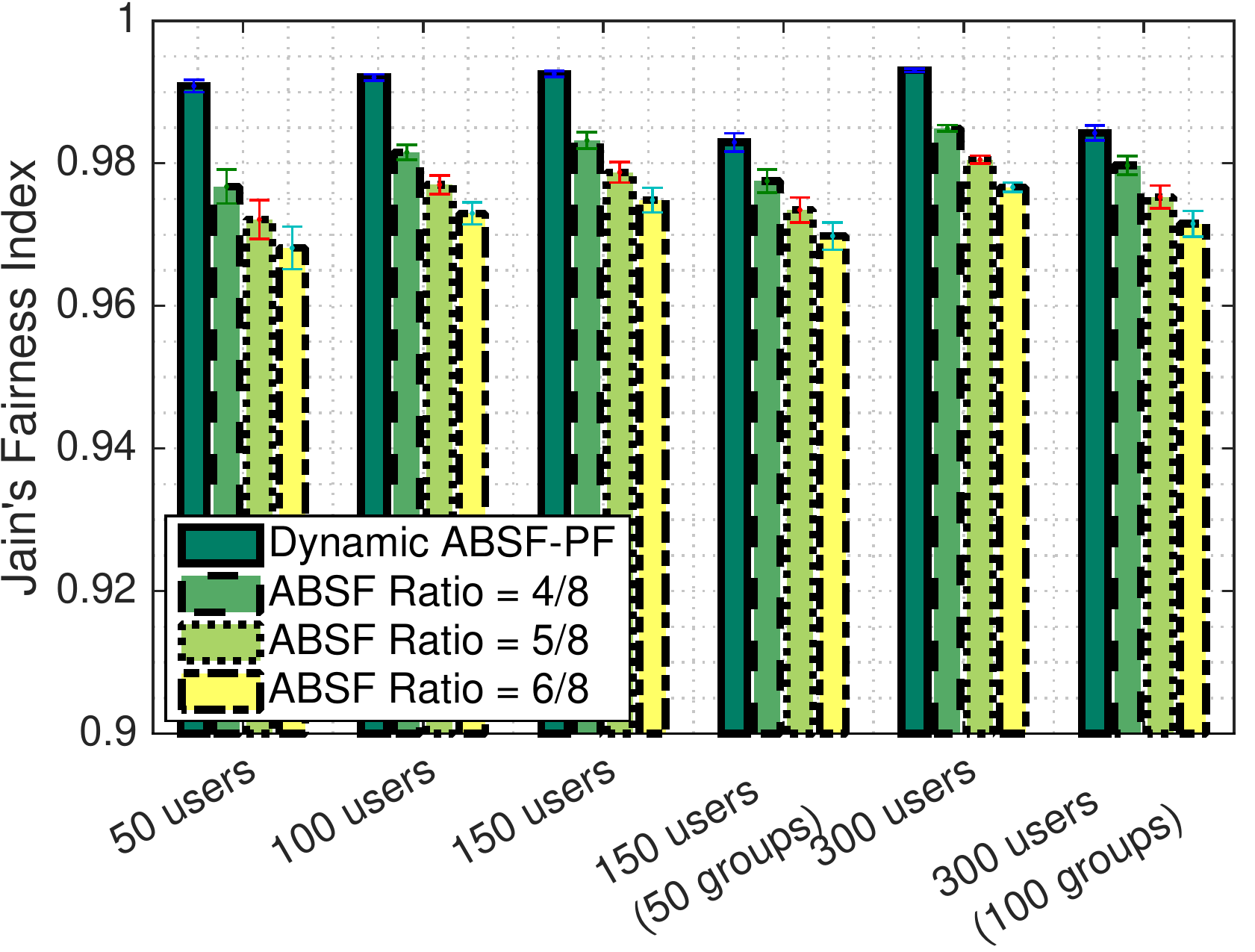}
%\vspace{-4mm}
	\caption{\footnotesize Jain's Fairness Index under different ABS schemes.}
	\label{fig:fairness_stocastic}
%\vspace{-5mm}
\end{figure}
\fi

In summary, our stochastic ABS scheme outperforms current standard solutions and offers a bargain trade-off between user fairness and spectral efficiency. \texttt{Asymptotic ABS-PF} and \texttt{Dynamic ABS-PF} provide extremely high fairness levels preserving reasonable throughput values, comparable to the maximum achievable when neglecting fairness issues. As final remark, exploiting mmD2D relay communications is a notable advantage of our scheme.  

%In general, while ABS with fixed application ratios provide good level of fairness and sufficient throughput levels, our \texttt{ABS-PF} scheme outperforms the current standard solutions and shows interesting potential gain.
% of applying enhanced ABS schemes to optimally balance user fairness and overall system throughput. 

%\vspace{2mm}
\subsection{Performance in a Heterogeneous Deployment}
Since homogeneous base station deployments might bias our results, we next take into consideration a realistic scenario, focusing our attention on a particular use case: London city, as previously presented in Section~\ref{s:example} for the area reported in Fig.~\ref{fig:london_deployment}. 
%{\color{red} We take a dense-urban area of $400$m~$\times$~$320$m close to Oxford Circus metro station. Based on its real topography, we place only the base stations under the control of O2 mobile network operator.}\!\footnote{All information are retrieved from the OFCOM reports, available at \url{http://stakeholders.ofcom.org.uk/sitefinder/sitefinder-dataset/}
%, as other operators use other frequencies and, thus, they do not interfere with such base stations.
%}
%{\color{red} In this area, $9$ base stations are built with their real transmission parameters, as illustrated in Fig.~\ref{fig:london_deployment}.}
%
%\begin{figure} [!t]
%\centering
%	\includegraphics[width=0.30\textwidth]{figs/london_picture.png}
%	\vspace{-3mm}
%	\caption{\footnotesize O2 deployment in London city - Oxford Circus.}
%	\label{fig:london_deployment}
%	\vspace{-3mm}
%\end{figure}
%
We consider a total amount of $1000$ users, as the user density for such area is $\sim \! 8000$ users/km~$\!\!^2$ and the area considered is $0.128$ km~$\!\!^2\!$~\cite{london_report}.
Users are placed within the considered area following two distributions: one guides users mobility behavior along the streets while the other characterizes the static user positions when they are within the buildings. Along our simulations, we vary the ratio $\Delta$ between the averages of those two distributions to model different day-time periods.
Mobile users follow a constrained RWP model: $(i)$ they select randomly a speed and a destination location within a valid street of the map, $(ii)$ then they follow the shortest path and reach the new destination by following the streets of the map. When relay groups are in place, the group, e.g., its center of gravity, follows the mobility model rules on the streets while the users of the group are randomly placed around the center of gravity. Users within the buildings are statically allocated at random according to a uniform distribution.

%Similarly to what presented in Section~\ref{ss:optimal_selection}, we compare our practical approach with \texttt{Legacy}, \texttt{BSB}, and \texttt{Max Throughput}, and, due to lack of space, we omit results for ABS with fixed application ratios, that are however aligned with the results presented for the homogeneous case. 
%However, since we have shown in Sections~\ref{ss:optimal_selection} that D2D relay is always beneficial, 
%Here we consider two benchmarking schemes in which D2D groups can be formed also under \texttt{Legacy} and \texttt{BSB}. We refer to the resulting schemes as \texttt{Legacy-D2D} and \texttt{BSB-D2D}. 
Since we have shown in Sections~\ref{ss:optimal_selection} that mmD2D relay is always beneficial, here we consider two benchmarking schemes in which mmD2D groups are also present. In the first scheme, we allow the formation of mmD2D groups also under \texttt{BSB}, while in the second scheme we use the \texttt{DRONEE} mechanism defined in~\cite{asadi2014dronee} to form clusters dynamically, under legacy base station operation (no ABS). We refer to the resulting schemes as \texttt{BSB-D2D} and \texttt{DRONEE}, respectively.
We have carried out different simulations to evaluate the realistic deployment in different operational timeframes. For each timeframe, we properly model the ratio $\Delta$ between the distribution of users moving along the streets and the users staying within the buildings in the following way: $(i)$ {\it Peak Hours}, during lunch time, $\Delta = 70\!:\!30$, $(ii)$ {\it Business Hours}, during morning and afternoon, $\Delta = 40\!:\!60$, $(iii)$ {\it Night Hours}, $\Delta = 10\!:\!90$~\cite{london_report}.

In Fig.~\ref{fig:london_throughput}, we show the system throughput expressed for different schemes. We observe that the system throughput increases during peak hours, as most of the people are moving outside and, thus, exploiting opportunistic relay over mmD2D sidelinks brings an additional gain. \texttt{Asymptotic ABS-PF} and \texttt{Dynamc ABS-PF} perform quite well, showing similar throughput figures as \texttt{DRONEE} and significantly worse throughput results only if compared with \texttt{Max Throughput}. Note that \texttt{Asymptotic ABS-PF} obtains only slightly lower throughput than \texttt{Dynamic ABS-PF}, although it results in much lower fairness. 
More in detail, due to the heterogeneity of the realistic scenario, user fairness is significantly impaired compared to results obtained for a homogeneous deployment, as shown in Fig.~\ref{fig:london_fairness}. However, \texttt{Dynamic ABS-PF} unveils the great potentials of properly applying a dynamic ABS scheme, outperforming not only the asymptotic optimization scheme, but also \texttt{Max Throughput} and \texttt{DRONEE} solutions in terms of JFI by about $100\%$ and $160\%$, respectively. 
Moreover, the fairness achieved with \texttt{BSB-D2D} is comparable with or better than our proposal's one, though it provides much less throughput.

\ifsingle
\begin{figure} [!t]
\centering
	\includegraphics[scale=0.47]{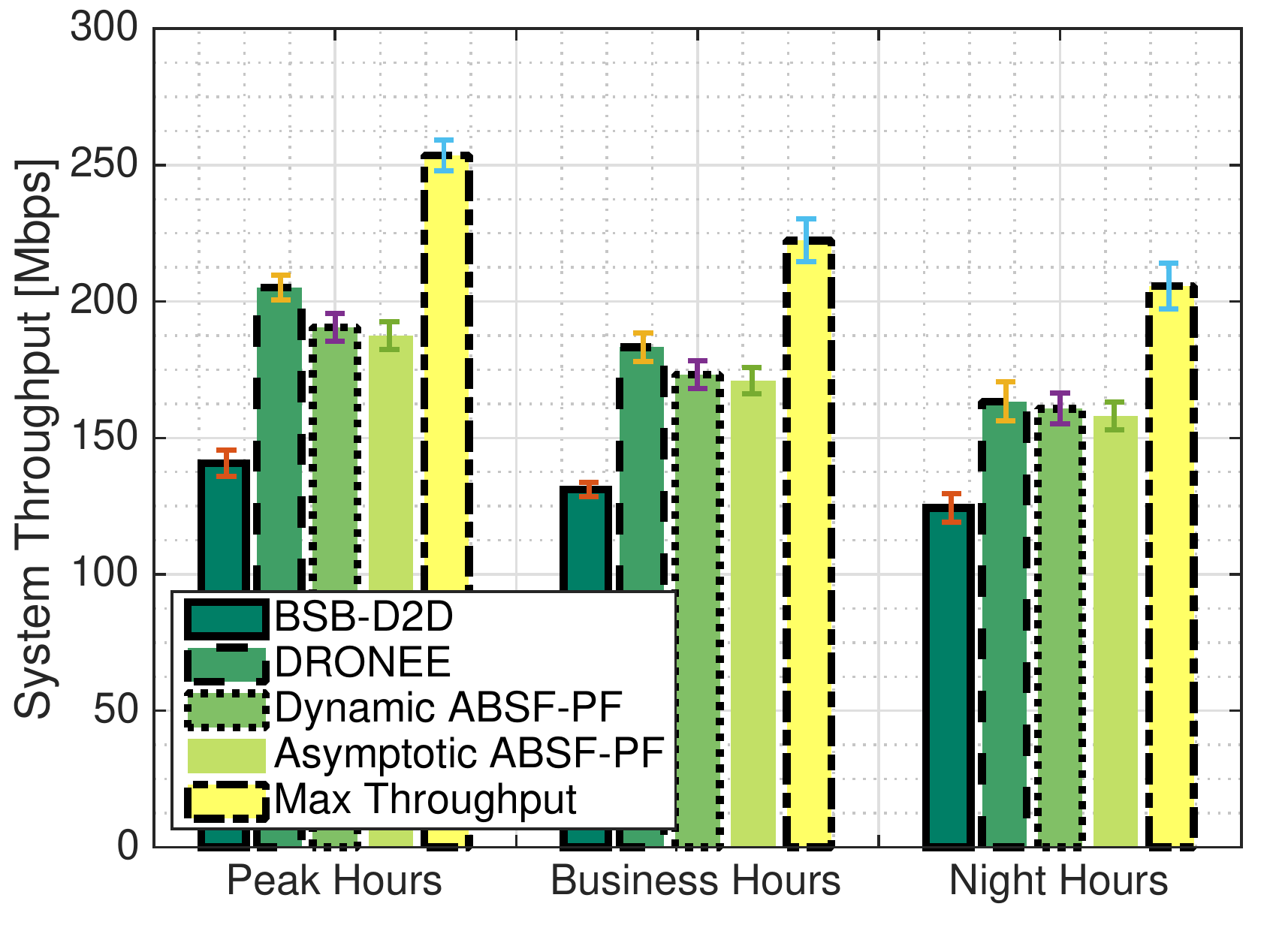}
%\vspace{-3mm}
	\caption{\footnotesize System throughput in London deployment.}
	\label{fig:london_throughput}
%\vspace{-3mm}
\end{figure}

\begin{figure} [!t]
\centering
	\includegraphics[scale=0.47]{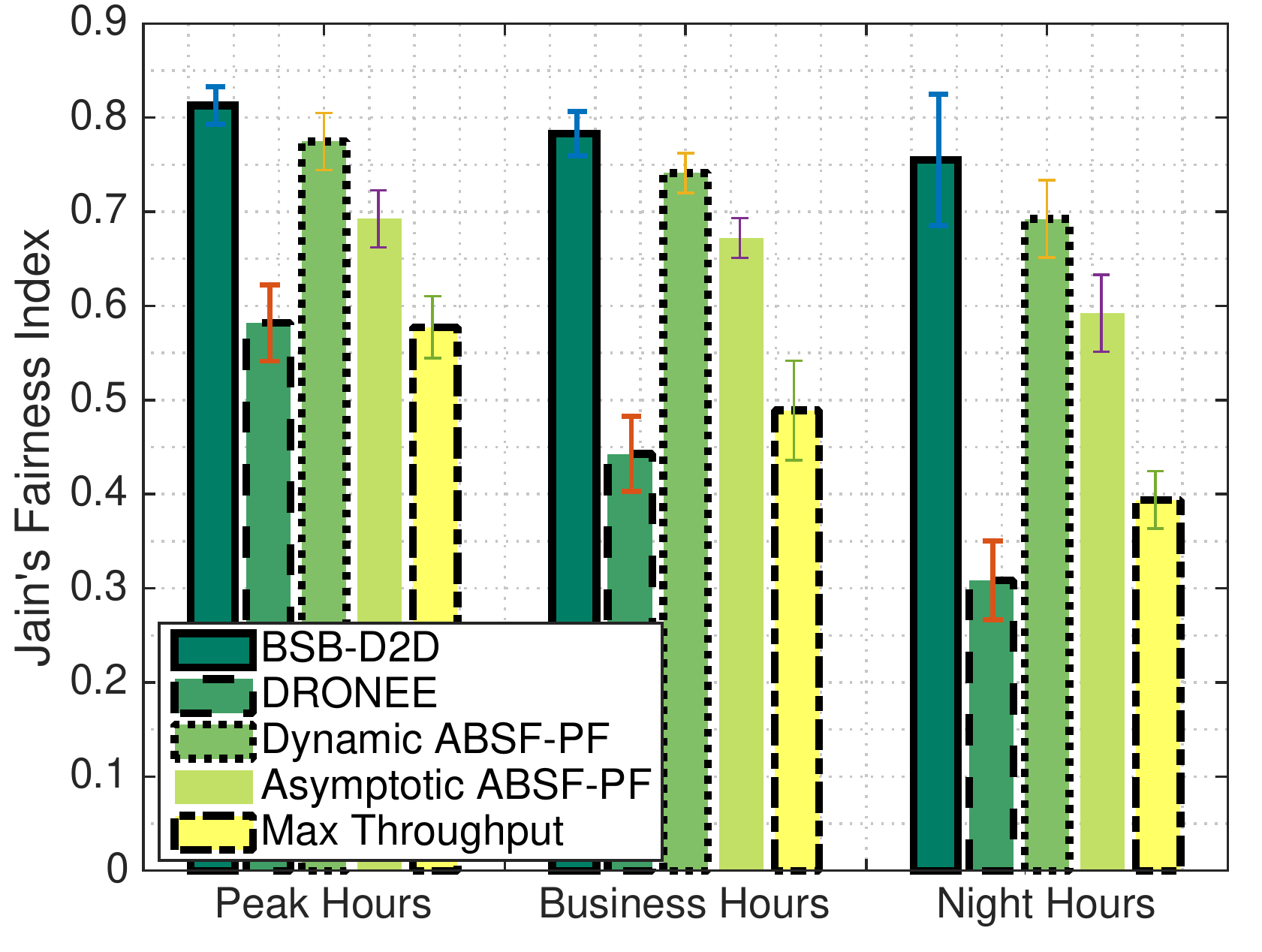}
%\vspace{-4mm}
	\caption{\footnotesize Jain's Fairness Index in London deployment.}
	\label{fig:london_fairness}
%\vspace{-5mm}
\end{figure}
\else
\begin{figure} [!t]
\centering
	\includegraphics[scale=0.31]{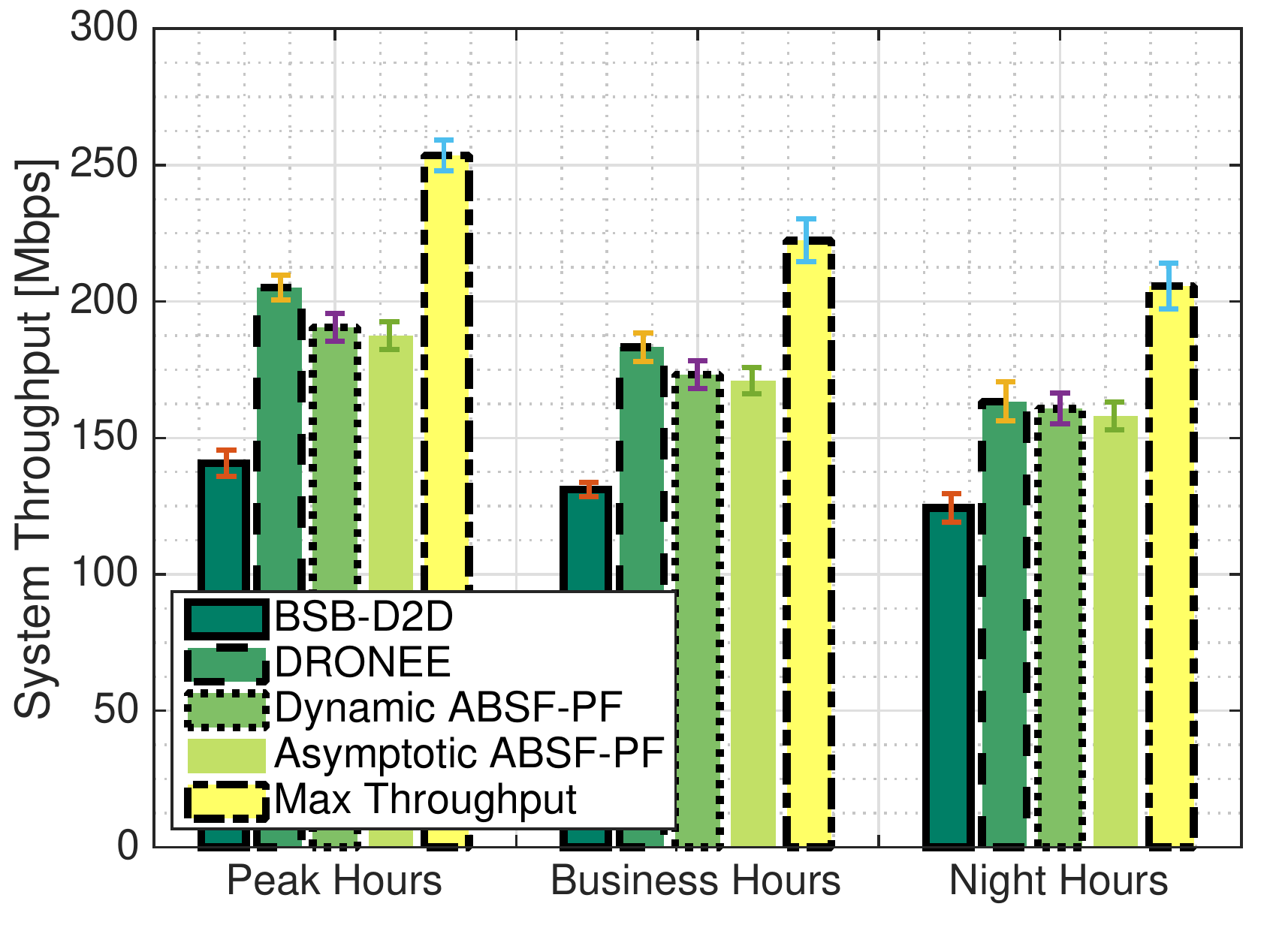}
%\vspace{-3mm}
	\caption{\footnotesize System throughput in London deployment.}
	\label{fig:london_throughput}
%\vspace{-3mm}
\end{figure}

\begin{figure} [!t]
\centering
	\includegraphics[scale=0.31]{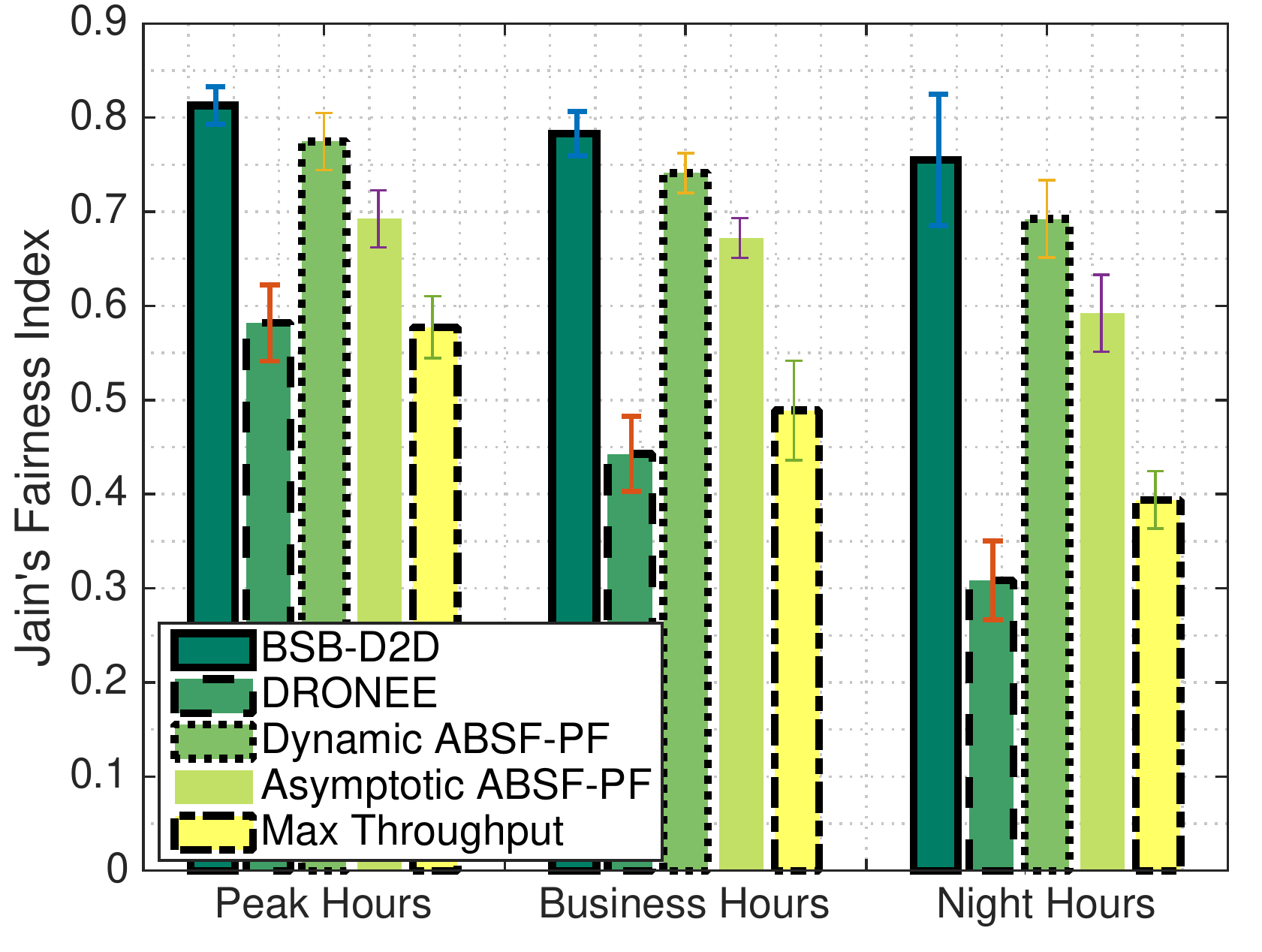}
%\vspace{-4mm}
	\caption{\footnotesize Jain's Fairness Index in London deployment.}
	\label{fig:london_fairness}
%\vspace{-5mm}
\end{figure}
\fi

In conclusion, our results illustrate how \texttt{Dynamic ABS-PF} is intrinsically better than standard approaches and pure (static) ABS optimizations targeting throughput or low interference. In particular, although we have shown that ABS alone is not able to boost throughput, \texttt{Dynamic ABS-PF} manages to handle the throughput enhancements achievable with mmD2D relay while achieving very high fairness levels.
However, in case of homogeneous scenarios, the importance of dynamic optimization becomes lower, and an asymptotic and less complex optimization approach can be used instead.

\section{Related work}
\label{s:related}

%
%
%
%\changes{Some recent papers to cite and comment here: 
%\begin{itemize}
%    \item ABS used for power saving~\cite{VSSC17}
%    \item ABS and power control on femto cells, to reduce the impact on macro~\cite{MSRTF17}
%    \item using ABS to enable the coexistence of unlicensed networks (LTE + LTE-U and WiFi)~\cite{gawlowicz2018enabling},  a patent is available on this topic~\cite{melodia2018method}. 
%    \item ABS pricing scheme \cite{li2018pricing}
%    \item compute ABS patters by means of machine learning, 10\% suboptimal results~\cite{lynch2018ensemble} 
%\end{itemize}
%\vspace{1cm}
%}

%Hereinafter, we comment on the related work about D2D and ICI coordination solutions separately, as no joint approaches have been investigated so far.

Recent studies on D2D communications show the feasibility of such schemes, including under opportunistic scheduling assumptions requiring control decisions at millisecond timescales~\cite{AMG16}. The authors of \cite{andreev2014ComMag} describe how D2D in unlicensed spectrum has a higher implementation opportunity because it requires minor changes in the existing standards. In \cite{karvounas2014ComMag}, it is shown that D2D can be deployed in LTE networks and groups can be formed with WiFi-Direct.
The authors of~\cite{asadi2014dronee} and \cite{AMG16} propose a practical opportunistic scheme and a protocol for D2D over LTE and WiFi networks. 
%Zhou {\it et al.}~\cite{zhou2013TVT} propose an optimal resource utilization for  multicast relaying with D2D clusters. They provide a closed-form expression for the probability distribution function of the optimal number of relays in a cluster, and an intra-cluster retransmission scheme. 
The authors of~\cite{bao2013infocom} propose a clustering technique to increase the network capacity in dense networks. Such approach, namely {\it Dataspotting},  tries to avoid to send twice the same content within the same geografical area. 
%As concerns potential performances, 
%\cite{doppler2013patent, seppala2011WNCN, zhou2013TVT}  focus on D2D clustering in cellular networks. T
%the authors of~\cite{doppler2013patent, seppala2011WNCN} 
%the authors of~\cite{doppler2013patent} 
%show via simulations that D2D groups increase the throughput by up to $66\%$ in comparison to legacy systems. 
The impact of D2D (with WiFi Direct) and ABS (with BSB) has been partially studied via simulation in~\cite{ASM15} under finite load assumptions. However, in that work D2D relay speed is comparable to cellular speed, and the authors conclude that D2D and ABS in combination can bring quite limited value added.    

Note that our proposed opportunistic D2D relay is novel with respect to state of the art solutions because it uses mmWave and does not create performance bottlenecks in clusters (or relay groups). Moreover, a compound analytic approach to D2D and ICIC is completely missing so far in the literature.

While D2D is attracting the attention of industrial players, ABS has already become popular due to its trade-off between performance improvement and low implementation complexity, as widely shown by \cite{ghosh2012heterogeneous}.
ABS has been proposed for throttling macro base station transmissions in presence of micro and pico cells. However, much more interesting results have been shown when ABS has  been adopted for all kind of cells. 
Deterministic ABS approaches like in~\cite{TWC_sciancalepore_2016} 
have shown how pre-computed time-patterns can lead ABS-enabled cellular systems to near-optimal working points. 
In~\cite{TWC_sciancalepore_2016}, the authors takle the ICIC problem by inspiring a heuristic solution which provides a near-optimal deterministic ABS pattern to schedule all required traffic, when content distribution systems are involved. 
Another interesting solution, such as~\cite{kamel2012performance}, deals with heterogeneous networks in which a macro base station coordinates the activity of small base stations to improve throughput performance when sharing a limited area. 
%For a limited traffic distribution, the  deterministic approach proposed in~\cite{cierny2013on} aims to determine the best blanked subframe density according to a given traffic distribution. 
%
More advanced solutions focus on the pattern reuse which directly guides the ABS activity pattern. In particular,~\cite{son2011utility} derives the best temporal pattern duration, given a set of chosen patterns to maximize the total user throughput. However, as proved in our work, while some scenarios may adversely impact on the system throughput, a pure throughput maximization can lead to highly unfair throughputs.
Lastly, many other solutions focus on the trade-off between throughput and fairness using different approaches. \cite{singh2014joint} and \cite{NiuGLSJV15}, for example, propose to apply user association and D2D multi-hop offloading, respectively, to achieve such a goal. Nevertheless, differently from our approach, many changes to the current cellular architecture should be done in order to implement such solutions.

Our work completely differentiates from the literature, since we are the first to analytically study and design a joint scheme to provide high spectral efficiency by leveraging cooperative D2D opportunistic communications using mmWave sidelinks, while at the same time adjusting user fairness by means of ABS. 

\section{Conclusions}
\label{s:conclusions}

We have realistically modeled the performance of a cellular network with ABS and D2D relay with mmWave sidelinks, whose availability if not only foreseen in 5G system, but they are already available in commercial off-the-shelf devices implementing IEEE 802.11ad specifications. 

We have shown that interference coordination and transmission efficiency---and hence throughput and fairness---can be addressed simultaneously by stochastically coordinating gNBs with ABS and by leveraging opportunistic mmWave outband sidelinks. The analysis shows that ABS and relay operate orthogonally, so that they can be optimized separately. 

%Our analysis shows that legacy ABS is not able to provide fairness and throughput guarantees. In fact, it can only guarantee cleaner and more efficient channel transmissions. However, we have shown that ABS can be used to tune fairness in both static and dynamic topologies of users. To boost throughput while blanking subframes with ABS, the adoption of mmD2D relay sidelinks is key and it can be implemented independently of ABS parameters. Indeed, mmD2D results in improved throughput under any ABS state.

Hence, we have studied proportional fairness optimization problems in which the existence of sidelinks available for relay is an input, whereas ABS state probabilities are decision variables. Further optimizing on relay groups and sidelinks would be possible but required to study user interactions and incentives, which goes beyond the communication technology. So we have decided to keep that aspect for future work.  

The two problems we have formulated tackle fairness targets that are respectively long-term (i.e., asymptotic performance for static and/or dense topologies) and dynamic (i.e., based on history and mobility of users). 
%
%A network controller, whose architecture and operation is out of the scope of this work, would then be responsible to timely enforce the optimal ABS state probabilities computed by solving the formulated problems.
%
We have validated our proposals by means of numerical simulations that cover uniform ideal scenarios as well as a very realistic urban scenarios. Our results show that the compound impact of {\it stochastic} ABS pattern and {\it opportunistic} relay over mmWave sidelinks is highly beneficial in terms of both throughput and fairness.

\bibliographystyle{IEEEtran}
%\bibliography{biblioinfocom2016}
{
\footnotesize
\bibliography{biblioinfocom2016}
}

\appendix
\label{appendix_a}

\subsection{Proof of Proposition~\ref{prop:CDF}}

Let us first compute the CDF of the SINR conditional to a given value of the noise r.v. $Z$. By definition, we have the following expression:  

\begin{align}
& F_{\gamma|Z=z}(x) = \text{Pr} \left\{ \frac{S}{z + \sum_{j=1}^k I_j} \le x \right\} \nonumber \\
& \quad = 
\! \int_0^{\infty} \! 
\! \int_0^{\infty} \!
\dots
\! \int_0^{\infty} \!
\text{Pr} \left\{ S \le x \left( z + \left. \sum_{j=1}^k I_j \right| I_j = a_j  \right) \right\} 
\nonumber \\
& \quad \cdot \prod_{j=1}^k f_{I_j}(a_j) \; da_1 da_2 \dots da_k 
\nonumber \\
& \quad = 
\! \int_0^{\infty} \! 
\! \int_0^{\infty} \!
\dots
\! \int_0^{\infty} \!
\text{Pr} \left\{ S \le x \left( z + \sum_{j=1}^k a_j  \right) \right\} 
\nonumber \\
& \quad \cdot 
\prod_{j=1}^k
\left(
	\lambda_j e^{-\lambda_j a_j}
\right)
\;
da_1 da_2 \dots da_k
\nonumber \\
& \quad = 
\! \int_0^{\infty} \! 
\! \int_0^{\infty} \!
\dots
\! \int_0^{\infty} \!
\left( 1 - e ^ {-\lambda_S \left( z + \sum_{j=1}^k a_j  \right ) x}\right)
\nonumber \\
& \quad \cdot 
\prod_{j=1}^k
\left(
	\lambda_j e^{-\lambda_j a_j}
\right)
\;
da_1 da_2 \dots da_k
\nonumber \\
& \quad = 
1 - e ^ {-\lambda_S z x} \prod_{j=1}^k \frac{\lambda_j}{\lambda_j + \lambda_S x}.
\end{align}

The CDF of the SINR is then computed by removing the condition on $Z$: 
\begin{align}
    F_{\gamma} (x)& = E \left[ F_{\gamma | Z } (x) \right] \nonumber \\
    & = E \left[ 1 - e ^ {-\lambda_S x Z} \prod_{j=1}^k \frac{\lambda_j}{\lambda_j + x\lambda_S }  \right] \nonumber \\
    & = 1 -  E \left[ e ^ {-\lambda_S x Z  } \right]  \prod_{j=1}^k \frac{\lambda_j}{\lambda_j + x\lambda_S }. 
\end{align} 

In the above expression, the term $E \left[ e ^ {-\lambda_S x Z } \right]$ is the LST of the noise power computed in $\lambda_S x$. Assuming AWGN noise with zero average and variance $\sigma^2 = N$, $Z$ is simply the square of a Gaussian r.v., so that the LST at $\lambda_S x \ge 0$ is easy to compute as follows:
\begin{align}
    E \left[ e ^ {-\lambda_S x Z } \right] & = \int_{-\infty}^{+\infty} e^{-\lambda_S x y^2}
    \frac{e^{-\frac{y^2}{2 \sigma^2} }} {\sqrt{2 \pi \sigma^2}} dy \nonumber \\
    & = \frac{1}{\sqrt{1 + 2 \lambda_S x \sigma^2}}.
\end{align}
Therefore, the expression of the CDF of $\gamma$ in presence of AWGN is 
\begin{align}
    \!\!\!F_{\gamma}(x) & = 1 - \frac{1}{\sqrt{1 + 2 \lambda_S N x}} \prod_{j=1}^k \frac{\lambda_j}{\lambda_j + x\lambda_S }, \quad \forall x \ge 0.
\end{align}

\end{document}